\documentclass[conference]{IEEEtran}
\IEEEoverridecommandlockouts
\usepackage{amsmath,amssymb,amsfonts}
\usepackage{graphicx}
\usepackage{textcomp}
\usepackage{soul}
\usepackage{xcolor}
\usepackage{amsthm}
\usepackage{bm}
\usepackage{bbm}

\usepackage{graphicx}
\usepackage{subcaption}   
\usepackage{tikz}
\usetikzlibrary{positioning}
\usepackage{svg}
\usepackage{algorithm}
\usepackage{algpseudocode}
\usepackage{url}
\usepackage{comment}
\newtheorem{proposition}{Proposition}
\newtheorem{corollary}{Corollary}
\newtheorem{definition}{Definition}
\newtheorem{remark}{Remark}
\usepackage{dcolumn}
\usepackage{caption}
\usepackage[compatibility=false]{caption}
\newtheorem{theorem}{Theorem}

\newtheorem{example}{Example}
\usepackage{empheq}
\usepackage{amsmath,amssymb,amsfonts}
\usepackage{xcolor}
\usepackage[letterpaper,
            top=0.75in,
            bottom=0.75in,
            left=0.75in,
            right=0.75in,
            includeheadfoot]{geometry}
\begin{document}

\title{Structural Sign Herdability in Temporal Networks: A Sufficient Condition via $\pi_p$-Graphs\\
\author{
    Pradeep~M and Twinkle~Tripathy%
    \thanks{Pradeep~M is a Research Scholar and Twinkle~Tripathy is with the faculty of the Department of Electrical Engineering, Indian Institute of Technology Kanpur, Uttar Pradesh, India 208016. Pradeep~M is the corresponding author (email: \{pradeepm, ttripathy\}@iitk.ac.in).}%
}

}

\maketitle

\begin{abstract}

In this paper, we study the herdability of temporally switching directed networks. A temporal network is modeled as a switched system with a fixed switching sequence, which imposes more restrictive herdability conditions than those of conventional switched systems. By exploiting the relationship between temporal walks and the entries of the controllability matrix, we derive sufficient conditions for herdability in temporal networks. We further show that the magnitude of edge weights influences the sign pattern of the controllability matrix, thereby affecting herdability. Consequently, herdability in temporal networks depends not only on the network topology and switching durations, but also on the magnitude of the edge weights. Motivated by this observation, we establish equivalent graph-theoretic conditions for structural sign ($\mathcal{SS}$) herdability in temporal networks. In particular, we introduce the union multigraph of temporal subsystems and propose the notion of a $\pi$-graph. We show that the existence of a $\pi_p$-graph, which is a temporally evolving $\pi$-graph, is sufficient to guarantee $\mathcal{SS}$ herdability. Illustrative examples are provided to demonstrate the proposed results.
\end{abstract}


\section{Introduction}\label{s1}
\label{sec:introduction}
\IEEEPARstart{H}{erdability}, a relaxed version of controllability, is the ability of the system to drive its state to the positive orthant of the state space in finite time. In most of the real-world applications, the desired states are often positive values \cite{ruf2018herdable}. For example, the simultaneous regulation of liquid levels in interconnected tank networks, control of chemical concentrations in reaction systems, coordination of robotic swarms, and population dynamics and opinion dynamics in networked systems sometimes share a common desired characteristic: the system state is required to reach or cross a threshold within the nonnegative orthant. In multi-agent coordination settings, this manifests as agents moving collectively in a common direction.  

One of the early contributions to the study of herdability is presented in \cite{ruf2018herdable}, where it was shown that input connectivity is a necessary condition for complete herdability. Subsequent works have examined herdability from the perspective of the network topology and sign properties. In \cite{meng2020leader}, the authors provide leader selection strategies for achieving herdability in structurally balanced signed graphs that were developed using graph walk-based analysis, and the results were further extended to weakly balanced signed graphs.

In \cite{she2019characterizing}, sufficient conditions were established to ensure the herdability of followers in a signed network using 1-walk and 2-walks from the leader node. Furthermore, authors in \cite{she2020characterizing} analyzed the controllable subspace of such systems using generalized equitable partitions.
Later,the authors in \cite{de2023herdability} investigated the herdability of clustering balanced directed graphs consisting of $k$ clusters and derived algebraic conditions for herdability in several graph structures, particularly tree graphs, under certain assumptions involving a single leader.

Recently, authors in \cite{xiang2026finite} studied optimal herdability control for hierarchical linear MAS using a finite horizon formulation. Game theoretic approaches were proposed in \cite{guo2025herdability,guo2026herdability}, where \cite{guo2025herdability} introduced a neighbour-based game control system with graph-theoretic herdability conditions, and \cite{guo2026herdability} proposed a Laplacian Game Control framework to analyze herdability using the herd graph approach.

 While the above analysis focuses on static networks, recently, the authors in \cite{shen2025herdability} analyzed the herdability in switching signed undirected networks using the union graph and graph partition. They also studied the herdability of simultaneously structurally balanced networks by studying the switching walk between the nodes. The authors showed that if the union graph $\bar{G}$, whose dynamics can be
 expressed as \(\dot{x}(t) = \bar{A} x(t) + \bar{B} \bar{u}(t),\label{eq:union_dynamics}
\) where \( \bar{A} = \sum_{k=1}^{N} A_k,\) \(\bar{B} = [B_1, \ldots, B_N] \in \mathbb{R}^{n \times D},
\quad
\bar{u}_{\delta(t)}(t) = [u_1(t), \ldots, u_D(t)]^\top \in \mathbb{R}^{D},D \triangleq \sum_{i=1}^{N} m_i
\) is herdable, then the switched network formed by the set of subsystems $(\mathcal{A}_i,\mathcal{B}_i)$ is herdable. But this is not always true in the case of a temporal network. In a temporal network, the switching sequence is predefined. For example, consider the switching network given in Fig.\ref{fig:three_row} whose dynamics are described by
\[
\dot{x}(t)=\mathcal{A}_{\delta(t)}x(t)+\mathcal{B}u(t),
\]
    where $\mathcal{A}_{\delta(t)} \in \mathbb{R}^{n\times n}$ denotes the adjacency matrix corresponding to the topology $\mathcal{G}_{\delta(t)}$. The switching signal $\delta(t)=i$ indicates that the subsystem $(\mathcal{A}_i,\mathcal{B})$ is active at time $t$. The corresponding adjacency matrices $\{\mathcal{A}_1,\mathcal{A}_2,\mathcal{A}_3\}$ are given, and the input matrix is $\mathcal{B}=[1~0~0~0]^{\top}$.

\begin{figure}[ht]

\begin{subfigure}{0.15\textwidth}
\centering
\tikzset{every picture/.style={scale=0.5pt}} 

\begin{tikzpicture}[x=0.75pt,y=0.75pt,yscale=-1,xscale=1]

\draw  [fill={rgb, 255:red, 0; green, 147; blue, 0 }  ,fill opacity=1 ][line width=0.75]  (116.29,62.2) .. controls (116.29,57.67) and (120.15,54) .. (124.91,54) .. controls (129.66,54) and (133.52,57.67) .. (133.52,62.2) .. controls (133.52,66.73) and (129.66,70.4) .. (124.91,70.4) .. controls (120.15,70.4) and (116.29,66.73) .. (116.29,62.2) -- cycle ;
\draw  [fill={rgb, 255:red, 204; green, 0; blue, 0 }  ,fill opacity=1 ][line width=0.75]  (54,124.55) .. controls (54,120.02) and (57.86,116.35) .. (62.61,116.35) .. controls (67.37,116.35) and (71.22,120.02) .. (71.22,124.55) .. controls (71.22,129.08) and (67.37,132.75) .. (62.61,132.75) .. controls (57.86,132.75) and (54,129.08) .. (54,124.55) -- cycle ;
\draw  [fill={rgb, 255:red, 204; green, 0; blue, 0 }  ,fill opacity=1 ][line width=0.75]  (171.78,126.2) .. controls (171.78,121.67) and (175.63,117.99) .. (180.39,117.99) .. controls (185.14,117.99) and (189,121.67) .. (189,126.2) .. controls (189,130.73) and (185.14,134.4) .. (180.39,134.4) .. controls (175.63,134.4) and (171.78,130.73) .. (171.78,126.2) -- cycle ;
\draw  [fill={rgb, 255:red, 204; green, 0; blue, 0 }  ,fill opacity=1 ][line width=0.75]  (113.21,188.8) .. controls (113.21,184.27) and (117.06,180.6) .. (121.82,180.6) .. controls (126.58,180.6) and (130.43,184.27) .. (130.43,188.8) .. controls (130.43,193.33) and (126.58,197) .. (121.82,197) .. controls (117.06,197) and (113.21,193.33) .. (113.21,188.8) -- cycle ;
\draw [line width=1.5]    (118.2,71.2) -- (73.89,113.63) ;
\draw [shift={(71,116.4)}, rotate = 316.24] [fill={rgb, 255:red, 0; green, 0; blue, 0 }  ][line width=0.08]  [draw opacity=0] (14.56,-6.99) -- (0,0) -- (14.56,6.99) -- (9.67,0) -- cycle    ;

\draw (119.6,26.6) node [anchor=north west][inner sep=0.75pt]    {$\mathbf{1}$};
\draw (31.2,115.8) node [anchor=north west][inner sep=0.75pt]    {$\mathbf{2}$};
\draw (200.4,115.4) node [anchor=north west][inner sep=0.75pt]    {$\mathbf{3}$};
\draw (118,210.6) node [anchor=north west][inner sep=0.75pt]    {$\mathbf{4}$};
\draw (59,68.4) node [anchor=north west][inner sep=0.75pt]    {$\mathbf{1}$};

\end{tikzpicture}
\caption{$\mathcal{G}_1(\mathcal{A}_1,\mathcal{B})$}
\end{subfigure}
\centering
\begin{subfigure}{0.15\textwidth}
\centering

\tikzset{every picture/.style={scale=0.5pt}} 
 
\begin{tikzpicture}[x=0.75pt,y=0.75pt,yscale=-1,xscale=1]

\draw  [fill={rgb, 255:red, 0; green, 147; blue, 0 }  ,fill opacity=1 ][line width=0.75]  (169.29,48.92) .. controls (169.29,44.39) and (173.15,40.72) .. (177.91,40.72) .. controls (182.66,40.72) and (186.52,44.39) .. (186.52,48.92) .. controls (186.52,53.45) and (182.66,57.12) .. (177.91,57.12) .. controls (173.15,57.12) and (169.29,53.45) .. (169.29,48.92) -- cycle ;
\draw  [fill={rgb, 255:red, 204; green, 0; blue, 0 }  ,fill opacity=1 ][line width=0.75]  (107,111.27) .. controls (107,106.74) and (110.86,103.07) .. (115.61,103.07) .. controls (120.37,103.07) and (124.22,106.74) .. (124.22,111.27) .. controls (124.22,115.8) and (120.37,119.47) .. (115.61,119.47) .. controls (110.86,119.47) and (107,115.8) .. (107,111.27) -- cycle ;
\draw  [fill={rgb, 255:red, 204; green, 0; blue, 0 }  ,fill opacity=1 ][line width=0.75]  (224.78,111.91) .. controls (224.78,107.39) and (228.63,103.71) .. (233.39,103.71) .. controls (238.14,103.71) and (242,107.39) .. (242,111.91) .. controls (242,116.44) and (238.14,120.12) .. (233.39,120.12) .. controls (228.63,120.12) and (224.78,116.44) .. (224.78,111.91) -- cycle ;
\draw  [fill={rgb, 255:red, 204; green, 0; blue, 0 }  ,fill opacity=1 ][line width=0.75]  (166.21,175.52) .. controls (166.21,170.99) and (170.06,167.32) .. (174.82,167.32) .. controls (179.58,167.32) and (183.43,170.99) .. (183.43,175.52) .. controls (183.43,180.05) and (179.58,183.72) .. (174.82,183.72) .. controls (170.06,183.72) and (166.21,180.05) .. (166.21,175.52) -- cycle ;
\draw [line width=1.5]    (221,111.52) -- (133,112.28) ;
\draw [shift={(129,112.32)}, rotate = 359.5] [fill={rgb, 255:red, 0; green, 0; blue, 0 }  ][line width=0.08]  [draw opacity=0] (14.56,-6.99) -- (0,0) -- (14.56,6.99) -- (9.67,0) -- cycle    ;
\draw [line width=1.5]    (223.81,124.78) -- (181.4,168.32) ;
\draw [shift={(226.6,121.92)}, rotate = 134.25] [fill={rgb, 255:red, 0; green, 0; blue, 0 }  ][line width=0.08]  [draw opacity=0] (14.56,-6.99) -- (0,0) -- (14.56,6.99) -- (9.67,0) -- cycle    ;

\draw (172.8,15.4) node [anchor=north west][inner sep=0.75pt]    {$\mathbf{1}$};
\draw (84.4,104.6) node [anchor=north west][inner sep=0.75pt]    {$\mathbf{2}$};
\draw (253.6,104.2) node [anchor=north west][inner sep=0.75pt]    {$\mathbf{3}$};
\draw (171.2,199.4) node [anchor=north west][inner sep=0.75pt]    {$\mathbf{4}$};
\draw (169,85.12) node [anchor=north west][inner sep=0.75pt]    {$1$};
\draw (210,147.12) node [anchor=north west][inner sep=0.75pt]    {$1$};

\end{tikzpicture}

\caption{$\mathcal{G}_2(\mathcal{A}_2,\mathcal{B})$}
\end{subfigure}
\begin{subfigure}{0.15\textwidth}
\centering

\tikzset{every picture/.style={scale=0.5pt}} 
 
\begin{tikzpicture}[x=0.75pt,y=0.75pt,yscale=-1,xscale=1]

\draw  [fill={rgb, 255:red, 0; green, 147; blue, 0 }  ,fill opacity=1 ][line width=0.75]  (165.29,58.52) .. controls (165.29,53.99) and (169.15,50.32) .. (173.91,50.32) .. controls (178.66,50.32) and (182.52,53.99) .. (182.52,58.52) .. controls (182.52,63.05) and (178.66,66.72) .. (173.91,66.72) .. controls (169.15,66.72) and (165.29,63.05) .. (165.29,58.52) -- cycle ;
\draw  [fill={rgb, 255:red, 204; green, 0; blue, 0 }  ,fill opacity=1 ][line width=0.75]  (103,119.87) .. controls (103,115.34) and (106.86,111.67) .. (111.61,111.67) .. controls (116.37,111.67) and (120.22,115.34) .. (120.22,119.87) .. controls (120.22,124.4) and (116.37,128.07) .. (111.61,128.07) .. controls (106.86,128.07) and (103,124.4) .. (103,119.87) -- cycle ;
\draw  [fill={rgb, 255:red, 204; green, 0; blue, 0 }  ,fill opacity=1 ][line width=0.75]  (220.78,121.51) .. controls (220.78,116.99) and (224.63,113.31) .. (229.39,113.31) .. controls (234.14,113.31) and (238,116.99) .. (238,121.51) .. controls (238,126.04) and (234.14,129.72) .. (229.39,129.72) .. controls (224.63,129.72) and (220.78,126.04) .. (220.78,121.51) -- cycle ;
\draw  [fill={rgb, 255:red, 204; green, 0; blue, 0 }  ,fill opacity=1 ][line width=0.75]  (162.21,184.12) .. controls (162.21,179.59) and (166.06,175.92) .. (170.82,175.92) .. controls (175.58,175.92) and (179.43,179.59) .. (179.43,184.12) .. controls (179.43,188.65) and (175.58,192.32) .. (170.82,192.32) .. controls (166.06,192.32) and (162.21,188.65) .. (162.21,184.12) -- cycle ;
\draw [line width=1.5]    (119.4,131.12) -- (157.13,173.34) ;
\draw [shift={(159.8,176.32)}, rotate = 228.21] [fill={rgb, 255:red, 0; green, 0; blue, 0 }  ][line width=0.08]  [draw opacity=0] (14.56,-6.99) -- (0,0) -- (14.56,6.99) -- (9.67,0) -- cycle    ;

\draw (166,21.4) node [anchor=north west][inner sep=0.75pt]    {$\mathbf{1}$};
\draw (77.6,110.6) node [anchor=north west][inner sep=0.75pt]    {$\mathbf{2}$};
\draw (246.8,110.2) node [anchor=north west][inner sep=0.75pt]    {$\mathbf{3}$};
\draw (164.4,205.4) node [anchor=north west][inner sep=0.75pt]    {$\mathbf{4}$};
\draw (95,153.72) node [anchor=north west][inner sep=0.75pt]    {$-1$};

\end{tikzpicture}

\caption{$\mathcal{G}_3(\mathcal{A}_3,\mathcal{B})$}
\end{subfigure}
\caption{Temporally switching network  $\{\mathcal{A}_1 \rightarrow \mathcal{A}_2 \rightarrow \mathcal{A}_3\}$}
\label{fig:three_row}
\end{figure}
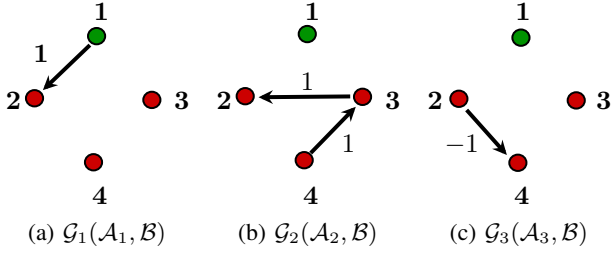

In the sense of switching networks, the system can be made herdable by selecting the switching sequence $\{\mathcal{A}_1 \rightarrow \mathcal{A}_3 \rightarrow \mathcal{A}_2\}$. However, in a temporal network, the order of switching is fixed. In this example, the predefined sequence is  $\{\mathcal{A}_1 \rightarrow \mathcal{A}_2 \rightarrow \mathcal{A}_3\}$, under which the network is not herdable. This example illustrates that a switching system that is herdable under an arbitrary switching sequence may not necessarily be herdable in the temporal setting where the switching order is predefined. The following scenarios highlight the motivations for the study of herdability in temporally switching networks.


\begin{itemize}
\item The union graph approach in \cite{shen2025herdability} may indicate herdability for the temporally switching network in Fig.~\ref{fig:three_row}, although the given switching sequence is not herdable. This motivates a separate analysis for temporally switching networks.
\item When multiple paths exist from the leader to a node whether within a single snapshot or across time their contributions to the controllability matrix may carry opposing signs, potentially cancelling exactly and yielding a net-zero contribution. Consequently, input-connectivity alone does not guarantee $\mathcal{SS}$ herdability, motivating a search for structural properties of temporal networks that guarantee it.
\item Our earlier work~\cite{pradeep2025structuralH} established that signed and layered dilation affects $\mathcal{SS}$ herdability in static networks. The coupling between time-varying dynamics and sign structure, however renders these notions considerably harder to analyze in the temporal setting, motivating a graph-theoretic condition, independent of dilations, that suffices for $\mathcal{SS}$ herdability.
\end{itemize}

 These observations motivate the study of herdability specifically for temporally switching networks. The  main contributions of this work are summarized as follows:
\begin{itemize}
    \item Leveraging the relationship between the controllability matrix and temporal walks, we derive necessary conditions and separately sufficient conditions for herdability in temporal networks where the underlying graph of each subsystem is a directed acyclic graph (DAG). 

    \item We analyze the effects of edge weights and snapshot durations on herdability in temporal networks, and investigate structural sign $\mathcal{SS}$ herdability under such temporal variations.

\item We introduce the union multigraph of temporal subsystems and derive a novel graph-theoretic condition, independent of signed and layer dilation, namely the existence of a spanning $\pi$-graph, that suffices to guarantee $\mathcal{SS}$ herdability.
    
\end{itemize}

\section{Notations and Background }\label{s2}
\subsection{Notation and Matrix Theory}\label{Notation}
The following notation is used throughout this paper: $[i,n] := {i,i+1,\ldots,n}$. For a vector $k$, $[k]i$ denotes its $i$th entry. A vector is said to be \emph{unisigned} if all its nonzero entries have the same sign. For a matrix $\mathcal{A}$, $\mathcal{A}{ij}$ denotes its $(i,j)$th entry, while $\mathcal{A}{(:,j)}$ and $\mathcal{A}{(i,:)}$ denote its $j$th column and $i$th row, respectively. The image and null space of $\mathcal{A}$ are denoted by $\mathrm{Im}(\mathcal{A}) := {y \mid y=\mathcal{A}v}$ and $\mathrm{Null}(\mathcal{A}) := {v \mid \mathcal{A}v=0}$, respectively. $\mathbb{R}^n_{>0}$ represents the positive orthant of the state space.
\subsection{Graph Theory}\label{Graph theory}

Let $\mathcal{G}=(\mathcal{V},\mathcal{E},\mathcal{A})$ be a weighted signed digraph, where $\mathcal{V}$ and $\mathcal{E}\subseteq\mathcal{V}\times\mathcal{V}$ denote the node and edge sets, respectively. An edge $(i,j)\in\mathcal{E}$ represents a directed edge from node $i$ to node $j$, with $a_{ij}\neq 0$ indicating $(j,i)\in\mathcal{E}$. A \emph{walk} in $\mathcal{G}(\mathcal{A},\mathcal{B})$ is a sequence of directed edges connecting an initial node to a final node, and its \emph{path product} is defined as the product of the corresponding edge weights.

\section{Problem statement}\label{s3}

\subsection{Linear Temporally Switching Systems}
A temporal network can be modelled as a switching system with a predefined switching order:  snapshot $i \in \{1,\dots,N\}$ corresponds to the system pair $(\mathcal{A}_i,\mathcal{B}_i)$ governing the network over $[t_{i-1},t_i)$ denoted as $T_i$, with dynamics
\begin{equation}\label{sys1}
\dot{\mathbf{x}}(t) = \mathcal{A}_i \mathbf{x}(t) + \mathcal{B}_i \mathbf{u}_i(t), \quad t \in [t_{i-1},t_i),
\end{equation}
where $\mathbf{x}(t) = [x_1(t),\dots,x_n(t)]^\top \in \mathbb{R}^n$ is the state vector. Here, $\mathcal{A}_i$ is the adjacency matrix of the $i^{th}$ snapshot, with $a^i_{jk} \neq 0$ iff $(k,j) \in \mathcal{E}_i$, and $\mathcal{B}_i$ specifies the driver nodes through which the control input $\mathbf{u}_i(t) \in \mathbb{R}^m$ is applied. In this paper, we investigate the herdability of temporal networks with a single leader node across all snapshots. Without loss of generality (WLOG), we assume node 1 to be the leader. Consequently, the input matrix is given by \(\mathcal{B}_i=B=b_ie_1\), where $e_1$ denotes the first standard basis vector in $\mathbb{R}^n$ and $b_i$ is the scalar input gain associated with snapshot $i$. WLOG, we assume $b_1 = 1$. For each subsystem $(\mathcal{A}_i,\mathcal{B})$, $i \in \{1,\dots,N\}$, the corresponding directed graph is denoted by $\mathcal{G}_i(\mathcal{A}_i,\mathcal{B})$.
        
      \vspace{-1mm}
\begin{figure}[ht]
    \centering
\tikzset{every picture/.style={line width=0.75pt}} 

\begin{tikzpicture}[x=0.65pt,y=0.55pt,yscale=-1,xscale=1]

\draw [color={rgb, 255:red, 0; green, 0; blue, 0 }  ,draw opacity=1 ] [dash pattern={on 4.5pt off 4.5pt}]  (303.67,188.76) .. controls (303.67,172.71) and (307.16,165.78) .. (306.59,151.36) ;
\draw [shift={(306.43,148.5)}, rotate = 85.88] [fill={rgb, 255:red, 0; green, 0; blue, 0 }  ,fill opacity=1 ][line width=0.08]  [draw opacity=0] (8.04,-3.86) -- (0,0) -- (8.04,3.86) -- (5.34,0) -- cycle    ;
\draw  [color={rgb, 255:red, 0; green, 0; blue, 0 }  ,draw opacity=1 ][fill={rgb, 255:red, 208; green, 2; blue, 27 }  ,fill opacity=1 ] (303.28,142.93) .. controls (303.28,141.05) and (304.82,139.52) .. (306.73,139.52) .. controls (308.64,139.52) and (310.18,141.05) .. (310.18,142.93) .. controls (310.18,144.82) and (308.64,146.35) .. (306.73,146.35) .. controls (304.82,146.35) and (303.28,144.82) .. (303.28,142.93) -- cycle ;
\draw [color={rgb, 255:red, 0; green, 0; blue, 0 }  ,draw opacity=0.67 ]   (209.35,206.48) -- (209.52,113.94) ;
\draw [shift={(209.52,110.94)}, rotate = 90.11] [fill={rgb, 255:red, 0; green, 0; blue, 0 }  ,fill opacity=0.67 ][line width=0.08]  [draw opacity=0] (8.04,-3.86) -- (0,0) -- (8.04,3.86) -- (5.34,0) -- cycle    ;
\draw [color={rgb, 255:red, 0; green, 0; blue, 0 }  ,draw opacity=0.67 ]   (181.92,186.06) -- (304.17,185.92) -- (351.98,185.87) ;
\draw [shift={(354.98,185.87)}, rotate = 179.93] [fill={rgb, 255:red, 0; green, 0; blue, 0 }  ,fill opacity=0.67 ][line width=0.08]  [draw opacity=0] (8.04,-3.86) -- (0,0) -- (8.04,3.86) -- (5.34,0) -- cycle    ;
\draw [color={rgb, 255:red, 0; green, 0; blue, 0 }  ,draw opacity=0.67 ]   (261.38,134.79) -- (196.85,198.68) ;
\draw [shift={(263.52,132.68)}, rotate = 135.29] [fill={rgb, 255:red, 0; green, 0; blue, 0 }  ,fill opacity=0.67 ][line width=0.08]  [draw opacity=0] (8.04,-3.86) -- (0,0) -- (8.04,3.86) -- (5.34,0) -- cycle    ;

\draw [color={rgb, 255:red, 0; green, 0; blue, 0 }  ,draw opacity=1 ][line width=0.75]  [dash pattern={on 4.5pt off 4.5pt}]  (201.22,209.83) .. controls (213.85,213.72) and (236.02,206.33) .. (239.64,207.51) .. controls (243.14,208.65) and (275.44,215.53) .. (294.75,199.79) ;
\draw [shift={(296.8,197.98)}, rotate = 136.43] [fill={rgb, 255:red, 0; green, 0; blue, 0 }  ,fill opacity=1 ][line width=0.08]  [draw opacity=0] (8.04,-3.86) -- (0,0) -- (8.04,3.86) -- (5.34,0) -- cycle    ;
\draw  [color={rgb, 255:red, 0; green, 0; blue, 0 }  ,draw opacity=1 ][fill={rgb, 255:red, 208; green, 2; blue, 27 }  ,fill opacity=1 ] (188.62,207.51) .. controls (188.62,205.65) and (190.38,204.14) .. (192.56,204.14) .. controls (194.73,204.14) and (196.49,205.65) .. (196.49,207.51) .. controls (196.49,209.37) and (194.73,210.88) .. (192.56,210.88) .. controls (190.38,210.88) and (188.62,209.37) .. (188.62,207.51) -- cycle ;
\draw  [color={rgb, 255:red, 0; green, 0; blue, 0 }  ,draw opacity=1 ][fill={rgb, 255:red, 208; green, 2; blue, 27 }  ,fill opacity=1 ] (134.26,141.26) .. controls (134.26,139.4) and (136.02,137.89) .. (138.19,137.89) .. controls (140.36,137.89) and (142.12,139.4) .. (142.12,141.26) .. controls (142.12,143.12) and (140.36,144.63) .. (138.19,144.63) .. controls (136.02,144.63) and (134.26,143.12) .. (134.26,141.26) -- cycle ;
\draw [color={rgb, 255:red, 0; green, 0; blue, 0 }  ,draw opacity=1 ][line width=0.75]  [dash pattern={on 4.5pt off 4.5pt}]  (143.21,147.07) .. controls (167.49,166.08) and (121.94,185.96) .. (186.63,206.88) ;
\draw [shift={(188.62,207.51)}, rotate = 197.37] [fill={rgb, 255:red, 0; green, 0; blue, 0 }  ,fill opacity=1 ][line width=0.08]  [draw opacity=0] (8.04,-3.86) -- (0,0) -- (8.04,3.86) -- (5.34,0) -- cycle    ;
\draw [shift={(153.54,188.63)}, rotate = 242.21] [fill={rgb, 255:red, 0; green, 0; blue, 0 }  ,fill opacity=1 ][line width=0.08]  [draw opacity=0] (8.04,-3.86) -- (0,0) -- (8.04,3.86) -- (5.34,0) -- cycle    ;
\draw  [color={rgb, 255:red, 0; green, 0; blue, 0 }  ,draw opacity=1 ][fill={rgb, 255:red, 208; green, 2; blue, 27 }  ,fill opacity=1 ] (296.87,195.15) .. controls (296.87,193.28) and (298.63,191.77) .. (300.8,191.77) .. controls (302.98,191.77) and (304.74,193.28) .. (304.74,195.15) .. controls (304.74,197.01) and (302.98,198.52) .. (300.8,198.52) .. controls (298.63,198.52) and (296.87,197.01) .. (296.87,195.15) -- cycle ;
\draw  [color={rgb, 255:red, 191; green, 231; blue, 148 }  ,draw opacity=0.14 ][fill={rgb, 255:red, 126; green, 211; blue, 33 }  ,fill opacity=0.09 ][line width=0.75]  (209.71,89.09) -- (353.81,89.09) -- (353.81,185.59) -- (209.71,185.59) -- cycle ;
\draw  [color={rgb, 255:red, 191; green, 231; blue, 148 }  ,draw opacity=0.14 ][fill={rgb, 255:red, 126; green, 211; blue, 33 }  ,fill opacity=0.09 ][line width=0.75]  (251.06,48.58) -- (395.16,48.58) -- (395.16,145.08) -- (251.06,145.08) -- cycle ;
\draw [color={rgb, 255:red, 191; green, 231; blue, 148 }  ,draw opacity=0.14 ][fill={rgb, 255:red, 126; green, 211; blue, 33 }  ,fill opacity=0.09 ][line width=0.75]    (209.71,89.09) -- (251.06,48.58) ;
\draw [color={rgb, 255:red, 191; green, 231; blue, 148 }  ,draw opacity=0.14 ][fill={rgb, 255:red, 126; green, 211; blue, 33 }  ,fill opacity=0.09 ][line width=0.75]    (209.71,185.59) -- (251.06,145.08) ;
\draw [color={rgb, 255:red, 191; green, 231; blue, 148 }  ,draw opacity=0.14 ][fill={rgb, 255:red, 126; green, 211; blue, 33 }  ,fill opacity=0.09 ][line width=0.75]    (353.81,185.59) -- (395.16,145.08) ;
\draw [color={rgb, 255:red, 191; green, 231; blue, 148 }  ,draw opacity=0.14 ][fill={rgb, 255:red, 126; green, 211; blue, 33 }  ,fill opacity=0.09 ][line width=0.75]    (353.81,89.09) -- (395.16,48.58) ;

\draw  [color={rgb, 255:red, 191; green, 231; blue, 148 }  ,draw opacity=0.14 ][fill={rgb, 255:red, 126; green, 211; blue, 33 }  ,fill opacity=0.09 ][line width=0.75]  (209.71,89.93) -- (251.06,48.58) -- (395.16,48.58) -- (395.16,145.07) -- (353.8,186.42) -- (209.71,186.42) -- cycle ; \draw  [color={rgb, 255:red, 191; green, 231; blue, 148 }  ,draw opacity=0.14 ][line width=0.75]  (395.16,48.58) -- (353.8,89.93) -- (209.71,89.93) ; \draw  [color={rgb, 255:red, 191; green, 231; blue, 148 }  ,draw opacity=0.14 ][line width=0.75]  (353.8,89.93) -- (353.8,186.42) ;

\draw (128.62,120.45) node [anchor=north west][inner sep=0.75pt]  [font=\footnotesize]  {$x( t_{0})$};
\draw (179.88,214.39) node [anchor=north west][inner sep=0.75pt]  [font=\footnotesize]  {$x( t_{1})$};
\draw (294.24,121.97) node [anchor=north west][inner sep=0.75pt]  [font=\footnotesize]  {$x( t_{f})$};
\draw (167.52,155.37) node [anchor=north west][inner sep=0.75pt]  [font=\scriptsize,color={rgb, 255:red, 28; green, 20; blue, 199 }  ,opacity=1 ]  {$u_{1}( t)$};
\draw (243.68,190.74) node [anchor=north west][inner sep=0.75pt]  [font=\scriptsize,color={rgb, 255:red, 28; green, 20; blue, 199 }  ,opacity=1 ]  {$u_{2}( t)$};
\draw (304.65,194.57) node [anchor=north west][inner sep=0.75pt]  [font=\footnotesize]  {$x( t_{2})$};
\draw (310.4,167.31) node [anchor=north west][inner sep=0.75pt]  [font=\scriptsize,color={rgb, 255:red, 28; green, 20; blue, 199 }  ,opacity=1 ]  {$u_{3}( t)$};
\draw (272.33,63.4) node [anchor=north west][inner sep=0.75pt]    {$\mathbb{R}^{n}  >0$};

\end{tikzpicture}
    \caption{Temporal herdability: the state is steered from $x(t_0)$ to $x(t_f) \in \mathbb{R}^n_{>0}$ using piecewise-defined inputs. The shaded region represents the positive orthant.}
    
    \label{fig:1}
\end{figure}
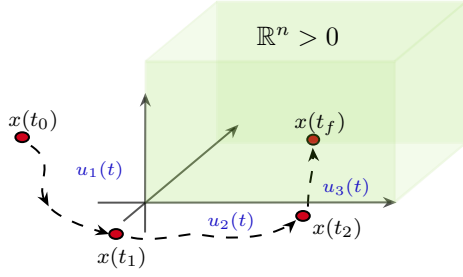

\begin{definition}
    Consider the temporal network governed by \eqref{sys1} over the temporal sequence 
    $\{T_i\}_{i=1}^N$, where  $T_i = [t_{i-1}, t_i)$ and $t_0 < t_1 < \cdots < t_N=t_f$. The system is said to be herdable over $\{T_i\}_{i=1}^N$ if, for every initial condition $x(t_0)$ $\in \mathbb{R}^n$, there exists piecewise continuous inputs $u_i : [t_{i-1}, t_i) \to \mathbb{R}^{m_i}$, $i = 1, \ldots, N$, such that $x(t_f) \in \mathbb{R}^n_{>0}$.
\end{definition}
Fig.~\ref{fig:1} illustrates temporal herdability, where the system state is steered from an arbitrary initial condition to the positive orthant using a piecewise input. Let us assume that the initial condition $x(0)=0$, then the reachable subspace of system $\eqref{sys1}$ is given in \cite{li2017fundamental} as

{ \small
{
\begin{equation}
\label{C-subspace}
\Omega_{\{ T_i\}}
=
\langle \mathcal{A}_N \mid \mathcal{B}_N \rangle
+
\sum_{j=2}^{N}
\left(
\prod_{i=j}^{N} e^{\mathcal{A}_i T_i}
\right)
\langle \mathcal{A}_{j-1} \mid \mathcal{B}_{j-1} \rangle
\end{equation}
}}
where $\langle \mathcal{A}_N \mid \mathcal{B} \rangle = \sum_{k=0}^{n-1} \mathcal{A}_N^{k}\,\mathcal{R}(\mathcal{B})$ denotes the controllable subspace associated with snapshot $N$, $\mathcal{R}(\mathcal{B})$ denotes the column space of $\mathcal{B}$, $T_i=[t_{i-1},t_i)$ denotes the $i$-th time interval, and 
$e^{\mathcal{A}_iT_i}=e^{\mathcal{A}_i(t_i-t_{i-1})}$ denotes the state transition matrix $(\mathcal{STM})$ associated with $\mathcal{A}_i$ over $T_i$.
\begin{remark}\cite{zhang2024reachability}
For the temporally switching network \eqref{sys1}, the controllability matrix $\mathcal{C}(t_f)$ from \eqref{C-subspace}, evaluated at time $t_f$, can be expressed as
\vspace{-1mm}
\begin{equation}\label{C-temp}
\mathcal{C}(t_f) = \begin{bmatrix} e^{\mathcal{A}_N T_N}\cdots e^{\mathcal{A}_2 T_2} C_1, \ldots,\; e^{\mathcal{A}_N T_N} C_{N-1},\; C_N \end{bmatrix},
\end{equation}
where $C_k$, $k \in \{1,\dots,N\}$, is the controllability matrix of the pair $(\mathcal{A}_k,\mathcal{B})$. Each term corresponds to a controllability subspace generated during an earlier switching interval and propagated through the subsequent dynamics; e.g., $e^{\mathcal{A}_N T_N} C_{N-1}$ represents the directions produced over interval $N-1$ after evolving under $\mathcal{A}_N$ for $[t_{N-1},t_N)$.
\end{remark}
 \begin{definition}[Temporal walk \cite{hou2016structural}] 

In a temporally switching network, a temporal walk from node $j$ to node $k$ is defined as a sequence of nodes $n_0, n_1, \ldots, n_L$, with $n_0 = j$ and $n_L = k$, together with snapshot indices $i_\ell \in \{1, \ldots, N\}$, satisfying $ i_1 \le i_2 \le \cdots \le i_L,$ such that $ a^{i_\ell}_{n_\ell,n_{\ell-1}} \neq 0, \quad \ell = 1, \ldots, L,$ where $L$ is the length of the walk.
 
\end{definition}

The path sign is the sign of the product of the edge weights along the walk. A node is said to be input-connected if there exists a temporal walk from a leader node to that node across the sequence of snapshots. In the subsequent section, we derive some conditions for herdability in a temporal network using temporal walks.

\section{Necessary and Sufficient Conditions for Herdability in Temporal Networks}\label{s4}
In this section, we derive necessary and sufficient conditions for herdability for a given realization of $(A_i, B)$. These conditions are obtained using the relationship between temporal walks and the entries of the controllability matrix, building on the results in \cite{zhang2024reachability}.

\begin{proposition}[\cite{pradeep2025structuralH}] \label{test for H}
A system is herdable if and only if there is no nonnegative vector $\bm{y}\in\mathbb{R}^n$ satisfying $\mathcal{C}(t_f)^{\top}\bm{y}=0$, where $\mathcal{C}(t_f)$ denotes the controllability matrix. Equivalently, if $\mathrm{Null}\big(\mathcal{C}(t_f)^{\top}\big)\cap\mathbb{R}{\geq 0}^n=\varnothing$, then there exists a vector $\bm{v}\in\mathbb{R}{>0}^n$ such that $\bm{v}\in\mathrm{Im}\big(\mathcal{C}(t_f)\big)$.

\end{proposition}

The following corollary applies the generalized theorem to characterize $\mathcal{SS}$ herdability in temporal networks.

\begin{corollary}\label{corollary1}
A linear temporally switching system \eqref{sys1} is herdable on the interval $[t_0,t_f)$ if and only if
$\mathrm{Null}(\mathcal{C}(t_f)^\top)\cap\mathbb{R}_{\geq0}^{\,n}=\{0\}$,
where $\mathcal{C}(t_f)$ is the controllability matrix of \eqref{sys1}, given by \eqref{C-temp}, evaluated at $t_f$. Equivalently, there exists $v>0$ such that $v\in\mathrm{Im}\big(\mathcal{C}(t_f)\big)$.
\end{corollary}
    
The above proposition and corollary provide conditions for assessing $\mathcal{SS}$ herdability for a given realization, beyond sign-pattern-based analysis. The following results present a necessary condition and separate sufficient conditions for herdability in temporal networks.
   
\begin{proposition}
    A temporally switching network is herdable only if all the nodes are input-connected by at least one temporal walk from the leader.
\end{proposition}

\begin{proof} Consider the temporally switching network governed by \eqref{sys1}. Let $\mathcal{C}(t_f)$ denote the controllability matrix defined in \eqref{C-temp}. Since the network is herdable, there exists a vector $\boldsymbol{\delta} \in \mathbb{R}^{nN}$ such that \(v=\mathcal{C}(t_f) \cdot \boldsymbol{\delta}>0\). Now expanding \(\mathcal{C}(t_f) \cdot \boldsymbol{\delta}\) gives \(\mathcal{C}(t_f)\cdot\boldsymbol{\delta}
= \sum_{j=1}^{nN} \delta_j, \mathcal{C}(t_f)_{(:,j)},\) where $\mathcal{C}(t_f)_{(:,j)}$ denotes the $j^{th}$ column of $\mathcal{C}(t_f)$. Hence $v$ is a linear combination of the columns of $\mathcal{C}(t_f)$.

Since $v_i>0$ for all $i \in \{1,\ldots,n\}$, each row $i$ of $\mathcal{C}(t_f)$ must contain at least one nonzero entry; otherwise the $i^{th}$ component of $\mathcal{C}(t_f)\cdot\boldsymbol{\delta}$ would be zero for any $\boldsymbol{\delta}$, contradicting $v_i>0$. Furthermore, each nonzero entry of the controllability matrix corresponds to the existence of a temporal walk from the leader node to the corresponding node in the network. Therefore, every node admits at least one temporal walk from the leader node. This implies that all the nodes are input-connected. Hence proved. 
\end{proof}

\begin{proposition}\label{prop:uni} Let $\mathcal{C}(t_f)$ be the controllability matrix associated with the temporally switching network given in \eqref{C-temp} evaluated at time $t_f$. If every row of $\mathcal{C}(t_f)$ contains at least one nonzero entry that belongs to a unisigned column, then the temporal network is herdable.
\end{proposition}
\begin{proof} Let $\boldsymbol{\delta} \in \mathbb{R}^{nN}$ be the preimage vector of the controllability matrix $\mathcal{C}(t_f)$ associated with the temporal network given in \eqref{C-temp}. By assumption, for every row $i$, $i \in \{1,2,\dots,n\}$, there exists a column  $j=j(i)$, $j \in \{1,2,\dots,nN\}$, such that $\mathcal{C}_{(i,j)} \ne 0 $ and column $j$ unisigned; i.e., all the nonzero entries of column $j$ have the same sign.

Let such entries be denoted by $\mathcal{C}(t_f)_{(i,j)}$. For each unisigned column $j$, define the corresponding entry of $\boldsymbol{\delta}$ such that \(\operatorname{sign}(\delta_j)=\operatorname{sign}(\mathcal{C}(t_f)_{(i,j)}),\) and set $\delta_j=0$ for all columns $j$ that are not unisigned. Since each unisigned column contributes positively to $\mathcal{C}(t_f)\cdot\boldsymbol{\delta}$, every component of $v=\mathcal{C}(t_f)\cdot\boldsymbol{\delta}$ is positive. Hence proved.
\end{proof} 
 \begin{proposition}
     If all the temporal walks to each node are unique and have the same path sign, then the network is herdable.
 \end{proposition}
 
\begin{proof}
Let $\mathcal{C}(t_f)$ be the controllability matrix associated with the temporally switching network given in \eqref{C-temp} evaluated at time $t_f$. By hypothesis, each node is reachable through a unique temporal walk; hence, every row of $\mathcal{C}$ contains at least one nonzero entry.

Moreover, since all temporal walks have the same path product sign, the corresponding columns of $\mathcal{C}$ are unisigned. Consequently, each row of $\mathcal{C}$  has a nonzero entry belonging to an unisigned column. Therefore, by Proposition~\ref{prop:uni}, the temporal network is herdable. \end{proof}

\begin{proposition}
  If the pair $(\mathcal{A}_N,\mathcal{B})$ is herdable, then the temporal network \eqref{sys1} is herdable for any $t \in T_N=[t_{f-1},t_f)$.
\end{proposition}

\begin{proof} The controllability matrix of the pair $(\mathcal{A}_N,\mathcal{B})$ is denoted by $\mathcal{C}_N$. If $\mathcal{C}_N$ is herdable, then there exists a vector $\boldsymbol{\bar{\delta}} \in \mathbb{R}^{n}$ such that \( \bar{v}=\mathcal{C}_N\cdot\boldsymbol{\bar{\delta}} >0 \).

Observe that $\mathcal{C}_N$ forms the last $n$ columns of the controllability matrix $\mathcal{C}(t_f)$ of the temporal network, namely the columns indexed by $[(N-1)n+1, nN]$. Define the vector $\boldsymbol{\delta}\in\mathbb{R}^{nN}$ as 
\[
\delta_l =
\begin{cases}
0, & 1 \le l \le (N-1)n \\
\bar{\delta}_{l-(N-1)n}, & (N-1)n+1 \le l \le nN .
\end{cases}
\]

With this construction, the expression $\mathcal{C}\cdot\boldsymbol{\delta}$ selects only the contribution from the last block of the columns $[(N-1)n+1,nN]$, yielding \(
\mathcal{C}(t_f)\cdot\boldsymbol{\delta} = \mathcal{C}(t_f)_N\cdot\boldsymbol{\bar{\delta}} = \bar{v} >0 .\) Therefore, the temporal network is herdable. 
\end{proof}

\begin{remark}
     Consider the controllability matrix of the temporal network given in \eqref{C-temp}. Let $j \in \{1,2,\dots,N\}$ denote the index of the blocks of $n$ columns in the controllability matrix. Specifically, $j=1$ corresponds to the first $n$ columns, i.e., columns $[1,n]$ and $j=2$ corresponds to the columns $[n+1,2n]$. Similarly, for $j=N$, the block corresponds to the last $n$ columns, namely $[(N-1)n+1,\,Nn]$. In general, the block of the controllability matrix corresponding to the $(p+1)^{\text{th}}$ snapshot is given by the columns $[pn+1,pn+n]$, where $p=0,1,\dots,N-1$.
\vspace{-1.5mm}
    \begin{equation} \label{C-expo}
\mathcal{C}(t_f) =
\begin{bmatrix}
\underbrace{\left(\prod_{k=j+1}^{N} e^{\mathcal{A}_r T_r}\right) C_j}_{j=1};
\underbrace{\left(\prod_{k=j+1}^{N} e^{\mathcal{A}_r T_r}\right) C_j}_{j=2};
..\;
\underbrace{C_N}_{j=N}
\end{bmatrix}
 \end{equation}
where $r=N-k+(j+1)$ and for $j=N$ the corresponding block is the controllability matrix of the pair $(\mathcal{A}_N,\mathcal{B})$. For $j=1,\dots,N-1$ each block corresponds to the controllability matrix of the pair $(\mathcal{A}_j,\mathcal{B})$ evolved through the subsequent snapshots through the corresponding $\mathcal{STM}$s given by \(\left(\prod_{k=j+1}^{N} e^{\mathcal{A}_r T_r}\right)\), where $r=N-k+(j+1)$ , depending on the value of $j$. For example, the first $n$ columns of the controllability matrix correspond to $j=1$. In this case,  $k=2,\dots,N$,  in the expression \(\left(\prod_{k=j+1}^{N} e^{\mathcal{A}_r T_r}\right)\) , results in \([e^{\mathcal{A}_N T_N} e^{\mathcal{A}_{N-1} T_{N-1}} \cdots e^{\mathcal{A}_2 T_2} C_1]\) where $r=N-k+(j+1)$. This represents the controllability matrix of the pair $(\mathcal{A}_1,\mathcal{B})$ evolved through the subsequent snapshots through $\mathcal{STM}s$ from snapshots $2$ to $N$.
\end{remark}

\begin{remark} \label{Remark-path}
    If $[\mathcal{A}_n^k \mathcal{A}_m^l]_{(i,j)} \neq 0$, then there exists a temporal walk from node $i$ to node $j$ of length $k+l$. In particular, there exists a node $q$ such that there is a walk of length $l$ from $i$ to $q$ in $\mathcal{G}_m$ and a walk of length $k$ from $q$ to $j$ in $\mathcal{G}_n$. 
\end{remark}

\begin{remark}\label{lem1}
    Consider a temporally switching network and let  $[pn+1,\ldots,pn+n]$ denote the block associated with the $(p+1)^{th}$ snapshot in the controllability matrix. The time-dependent entries in column $pn+1$ correspond to the set of nodes that are reachable from the leader node via a temporal walk across all snapshots. The entries in column $pn+2$ represent the set of nodes that are reachable from the nodes that are at path length $1$ to the leader in the $(p+1)^{th}$ snapshot. Similarly, subsequent columns represent nodes reachable through a temporal walk starting from nodes reached in the preceding step from the leader node.
\end{remark}

 \begin{example}
     For example, consider the temporally switching network in Fig.\ref{egb}. The evolution  of the controllability matrix of the temporal network across each snapshot is given by
 \end{example}
     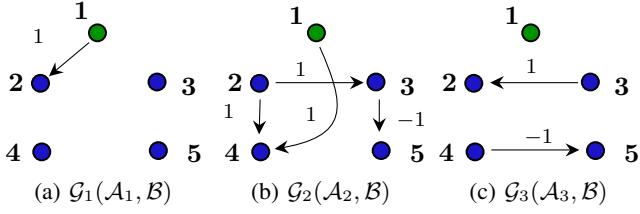
\begin{figure}[ht]

\begin{subfigure}{0.16\textwidth}
\centering
\tikzset{every picture/.style={scale=0.65pt}} 

\begin{tikzpicture}[x=0.75pt,y=0.75pt,yscale=-1,xscale=1]

\draw  [fill={rgb, 255:red, 0; green, 147; blue, 0 }  ,fill opacity=1 ][line width=0.75]  (189.29,72.69) .. controls (189.29,69.1) and (192.23,66.19) .. (195.84,66.19) .. controls (199.46,66.19) and (202.4,69.1) .. (202.4,72.69) .. controls (202.4,76.28) and (199.46,79.19) .. (195.84,79.19) .. controls (192.23,79.19) and (189.29,76.28) .. (189.29,72.69) -- cycle ;
\draw  [fill={rgb, 255:red, 28; green, 20; blue, 199 }  ,fill opacity=1 ][line width=0.75]  (145.29,111.69) .. controls (145.29,108.1) and (148.23,105.19) .. (151.84,105.19) .. controls (155.46,105.19) and (158.4,108.1) .. (158.4,111.69) .. controls (158.4,115.28) and (155.46,118.19) .. (151.84,118.19) .. controls (148.23,118.19) and (145.29,115.28) .. (145.29,111.69) -- cycle ;
\draw  [fill={rgb, 255:red, 28; green, 20; blue, 199 }  ,fill opacity=1 ][line width=0.75]  (146.29,164.69) .. controls (146.29,161.1) and (149.23,158.19) .. (152.84,158.19) .. controls (156.46,158.19) and (159.4,161.1) .. (159.4,164.69) .. controls (159.4,168.28) and (156.46,171.19) .. (152.84,171.19) .. controls (149.23,171.19) and (146.29,168.28) .. (146.29,164.69) -- cycle ;
\draw  [fill={rgb, 255:red, 28; green, 20; blue, 199 }  ,fill opacity=1 ][line width=0.75]  (235.29,110.69) .. controls (235.29,107.1) and (238.23,104.19) .. (241.84,104.19) .. controls (245.46,104.19) and (248.4,107.1) .. (248.4,110.69) .. controls (248.4,114.28) and (245.46,117.19) .. (241.84,117.19) .. controls (238.23,117.19) and (235.29,114.28) .. (235.29,110.69) -- cycle ;
\draw  [fill={rgb, 255:red, 28; green, 20; blue, 199 }  ,fill opacity=1 ][line width=0.75]  (236.29,163.69) .. controls (236.29,160.1) and (239.23,157.19) .. (242.84,157.19) .. controls (246.46,157.19) and (249.4,160.1) .. (249.4,163.69) .. controls (249.4,167.28) and (246.46,170.19) .. (242.84,170.19) .. controls (239.23,170.19) and (236.29,167.28) .. (236.29,163.69) -- cycle ;
\draw    (189.4,80.19) -- (160.69,104.26) ;
\draw [shift={(158.4,106.19)}, rotate = 320.01] [fill={rgb, 255:red, 0; green, 0; blue, 0 }  ][line width=0.08]  [draw opacity=0] (10.72,-5.15) -- (0,0) -- (10.72,5.15) -- (7.12,0) -- cycle    ;

\draw (175.6,46.6) node [anchor=north west][inner sep=0.75pt]    {$\mathbf{1}$};
\draw (125.2,102.8) node [anchor=north west][inner sep=0.75pt]    {$\mathbf{2}$};
\draw (258.4,105.4) node [anchor=north west][inner sep=0.75pt]    {$\mathbf{3}$};
\draw (123,156.6) node [anchor=north west][inner sep=0.75pt]    {$\mathbf{4}$};
\draw (263,157.6) node [anchor=north west][inner sep=0.75pt]    {$\mathbf{5}$};
\draw (144,68.12) node [anchor=north west][inner sep=0.75pt]  [font=\footnotesize]  {$1$};

\end{tikzpicture}

\caption{$\mathcal{G}_1(\mathcal{A}_1,\mathcal{B})$}
\end{subfigure}
\centering
\begin{subfigure}{0.15\textwidth}
\centering

\tikzset{every picture/.style={scale=0.65pt}} 

\begin{tikzpicture}[x=0.75pt,y=0.75pt,yscale=-1,xscale=1]

\draw  [fill={rgb, 255:red, 0; green, 147; blue, 0 }  ,fill opacity=1 ][line width=0.75]  (275.29,237.49) .. controls (275.29,233.9) and (278.23,230.99) .. (281.84,230.99) .. controls (285.46,230.99) and (288.4,233.9) .. (288.4,237.49) .. controls (288.4,241.08) and (285.46,243.99) .. (281.84,243.99) .. controls (278.23,243.99) and (275.29,241.08) .. (275.29,237.49) -- cycle ;
\draw  [fill={rgb, 255:red, 28; green, 20; blue, 199 }  ,fill opacity=1 ][line width=0.75]  (231.29,276.49) .. controls (231.29,272.9) and (234.23,269.99) .. (237.84,269.99) .. controls (241.46,269.99) and (244.4,272.9) .. (244.4,276.49) .. controls (244.4,280.08) and (241.46,282.99) .. (237.84,282.99) .. controls (234.23,282.99) and (231.29,280.08) .. (231.29,276.49) -- cycle ;
\draw  [fill={rgb, 255:red, 28; green, 20; blue, 199 }  ,fill opacity=1 ][line width=0.75]  (232.29,329.49) .. controls (232.29,325.9) and (235.23,322.99) .. (238.84,322.99) .. controls (242.46,322.99) and (245.4,325.9) .. (245.4,329.49) .. controls (245.4,333.08) and (242.46,335.99) .. (238.84,335.99) .. controls (235.23,335.99) and (232.29,333.08) .. (232.29,329.49) -- cycle ;
\draw  [fill={rgb, 255:red, 28; green, 20; blue, 199 }  ,fill opacity=1 ][line width=0.75]  (321.29,275.49) .. controls (321.29,271.9) and (324.23,268.99) .. (327.84,268.99) .. controls (331.46,268.99) and (334.4,271.9) .. (334.4,275.49) .. controls (334.4,279.08) and (331.46,281.99) .. (327.84,281.99) .. controls (324.23,281.99) and (321.29,279.08) .. (321.29,275.49) -- cycle ;
\draw  [fill={rgb, 255:red, 28; green, 20; blue, 199 }  ,fill opacity=1 ][line width=0.75]  (325.29,328.49) .. controls (325.29,324.9) and (328.23,321.99) .. (331.84,321.99) .. controls (335.46,321.99) and (338.4,324.9) .. (338.4,328.49) .. controls (338.4,332.08) and (335.46,334.99) .. (331.84,334.99) .. controls (328.23,334.99) and (325.29,332.08) .. (325.29,328.49) -- cycle ;
\draw    (238.4,288.99) -- (237.69,312.98) ;
\draw [shift={(237.6,315.98)}, rotate = 271.69] [fill={rgb, 255:red, 0; green, 0; blue, 0 }  ][line width=0.08]  [draw opacity=0] (10.72,-5.15) -- (0,0) -- (10.72,5.15) -- (7.12,0) -- cycle    ;
\draw    (250.4,275.99) -- (315.4,275.99) ;
\draw [shift={(318.4,275.99)}, rotate = 180] [fill={rgb, 255:red, 0; green, 0; blue, 0 }  ][line width=0.08]  [draw opacity=0] (10.72,-5.15) -- (0,0) -- (10.72,5.15) -- (7.12,0) -- cycle    ;
\draw    (330.4,288.99) -- (330.4,310.99) ;
\draw [shift={(330.4,313.99)}, rotate = 270] [fill={rgb, 255:red, 0; green, 0; blue, 0 }  ][line width=0.08]  [draw opacity=0] (10.72,-5.15) -- (0,0) -- (10.72,5.15) -- (7.12,0) -- cycle    ;
\draw    (284.4,248.38) .. controls (304.49,294.09) and (308.29,318.06) .. (252.2,326.79) ;
\draw [shift={(249.6,327.18)}, rotate = 351.87] [fill={rgb, 255:red, 0; green, 0; blue, 0 }  ][line width=0.08]  [draw opacity=0] (10.72,-5.15) -- (0,0) -- (10.72,5.15) -- (7.12,0) -- cycle    ;

\draw (258.6,215.4) node [anchor=north west][inner sep=0.75pt]    {$\mathbf{1}$};
\draw (211.2,267.6) node [anchor=north west][inner sep=0.75pt]    {$\mathbf{2}$};
\draw (344.4,270.2) node [anchor=north west][inner sep=0.75pt]    {$\mathbf{3}$};
\draw (209,321.4) node [anchor=north west][inner sep=0.75pt]    {$\mathbf{4}$};
\draw (352,322.4) node [anchor=north west][inner sep=0.75pt]    {$\mathbf{5}$};
\draw (208,290.92) node [anchor=north west][inner sep=0.75pt]  [font=\footnotesize]  {$1$};
\draw (263.6,258.52) node [anchor=north west][inner sep=0.75pt]  [font=\footnotesize]  {$1$};
\draw (342,296.92) node [anchor=north west][inner sep=0.75pt]  [font=\footnotesize]  {$-1$};

\draw (271.6,293.32) node [anchor=north west][inner sep=0.75pt]  [font=\footnotesize]  {$1$};

\end{tikzpicture}

\caption{$\mathcal{G}_2(\mathcal{A}_2,\mathcal{B})$}
\end{subfigure}
\begin{subfigure}{0.16\textwidth}
\centering

\tikzset{every picture/.style={scale=0.65pt}} 

\begin{tikzpicture}[x=0.75pt,y=0.75pt,yscale=-1,xscale=1]

\draw  [fill={rgb, 255:red, 0; green, 147; blue, 0 }  ,fill opacity=1 ][line width=0.75]  (457.29,257.69) .. controls (457.29,254.1) and (460.23,251.19) .. (463.84,251.19) .. controls (467.46,251.19) and (470.4,254.1) .. (470.4,257.69) .. controls (470.4,261.28) and (467.46,264.19) .. (463.84,264.19) .. controls (460.23,264.19) and (457.29,261.28) .. (457.29,257.69) -- cycle ;
\draw  [fill={rgb, 255:red, 28; green, 20; blue, 199 }  ,fill opacity=1 ][line width=0.75]  (413.29,296.69) .. controls (413.29,293.1) and (416.23,290.19) .. (419.84,290.19) .. controls (423.46,290.19) and (426.4,293.1) .. (426.4,296.69) .. controls (426.4,300.28) and (423.46,303.19) .. (419.84,303.19) .. controls (416.23,303.19) and (413.29,300.28) .. (413.29,296.69) -- cycle ;
\draw  [fill={rgb, 255:red, 28; green, 20; blue, 199 }  ,fill opacity=1 ][line width=0.75]  (414.29,349.69) .. controls (414.29,346.1) and (417.23,343.19) .. (420.84,343.19) .. controls (424.46,343.19) and (427.4,346.1) .. (427.4,349.69) .. controls (427.4,353.28) and (424.46,356.19) .. (420.84,356.19) .. controls (417.23,356.19) and (414.29,353.28) .. (414.29,349.69) -- cycle ;
\draw  [fill={rgb, 255:red, 28; green, 20; blue, 199 }  ,fill opacity=1 ][line width=0.75]  (503.29,295.69) .. controls (503.29,292.1) and (506.23,289.19) .. (509.84,289.19) .. controls (513.46,289.19) and (516.4,292.1) .. (516.4,295.69) .. controls (516.4,299.28) and (513.46,302.19) .. (509.84,302.19) .. controls (506.23,302.19) and (503.29,299.28) .. (503.29,295.69) -- cycle ;
\draw  [fill={rgb, 255:red, 28; green, 20; blue, 199 }  ,fill opacity=1 ][line width=0.75]  (507.29,348.69) .. controls (507.29,345.1) and (510.23,342.19) .. (513.84,342.19) .. controls (517.46,342.19) and (520.4,345.1) .. (520.4,348.69) .. controls (520.4,352.28) and (517.46,355.19) .. (513.84,355.19) .. controls (510.23,355.19) and (507.29,352.28) .. (507.29,348.69) -- cycle ;
\draw    (435.4,296.19) -- (500.4,296.19) ;
\draw [shift={(432.4,296.19)}, rotate = 0] [fill={rgb, 255:red, 0; green, 0; blue, 0 }  ][line width=0.08]  [draw opacity=0] (10.72,-5.15) -- (0,0) -- (10.72,5.15) -- (7.12,0) -- cycle    ;
\draw    (433.4,348.19) -- (498.4,348.19) ;
\draw [shift={(501.4,348.19)}, rotate = 180] [fill={rgb, 255:red, 0; green, 0; blue, 0 }  ][line width=0.08]  [draw opacity=0] (10.72,-5.15) -- (0,0) -- (10.72,5.15) -- (7.12,0) -- cycle    ;

\draw (441.6,235.6) node [anchor=north west][inner sep=0.75pt]    {$\mathbf{1}$};
\draw (393.2,287.8) node [anchor=north west][inner sep=0.75pt]    {$\mathbf{2}$};
\draw (526.4,290.4) node [anchor=north west][inner sep=0.75pt]    {$\mathbf{3}$};
\draw (391,341.6) node [anchor=north west][inner sep=0.75pt]    {$\mathbf{4}$};
\draw (534,342.6) node [anchor=north west][inner sep=0.75pt]    {$\mathbf{5}$};
\draw (458,276.12) node [anchor=north west][inner sep=0.75pt]  [font=\footnotesize]  {$1$};
\draw (457,330.12) node [anchor=north west][inner sep=0.75pt]  [font=\footnotesize]  {$-1$};

\end{tikzpicture}
\caption{$\mathcal{G}_3(\mathcal{A}_3,\mathcal{B})$}
\end{subfigure}
\caption{Temporally switching network  $\{\mathcal{A}_1 \rightarrow \mathcal{A}_2 \rightarrow \mathcal{A}_3\}$}
\label{egb}
\end{figure}

For snapshot $3$, the entry $\mathcal{C}_3(5,2) = -\frac{T_2^2}{2} - T_3 T_2$ in the controllability matrix given below reflects two temporal walks from node $2$ to node $5$: $\mathcal{G}_2:\{2 \rightarrow 3 \rightarrow 5\}=-\frac{T_2^2}{2}$ and $\mathcal{G}_2:\{2 \rightarrow 4\} \rightarrow \mathcal{G}_3:\{4 \rightarrow 5\}=- T_2 T_3$.
\vspace{-2mm}
{ \small
\[C_3 =
\begin{bmatrix}
1 & 0 & \cdots & 1 & 0 & \cdots & 1 &  \cdots \\
0 & T_2 T_3 + 1 & \cdots & 0 & 0 & \cdots & 0  & \cdots \\
0 & T_2 & \cdots & 0 & 0 & \cdots & 0  & \cdots \\
T_2 & T_2 & \cdots & 0 & 1 & \cdots & 0  & \cdots \\
- T_2 T_3 & -\frac{T_2^2}{2} - T_3 T_2 & \cdots & 0 & -T_3 & \cdots & 0  & \cdots
\end{bmatrix}
\]
}
 Note that the pair $(\mathcal{A}_1,\mathcal{B})$ is active at snapshot $1$, and since the network evolves temporally, the weight of each edge depends on the duration of the snapshot in which it appears. Specifically, $\mathcal{G}_2:\{2 \rightarrow 3 \rightarrow 5\}$ is a length-$2$ path occurring entirely within snapshot $\mathcal{G}_2$, contributing weight $-\frac{T_2^2}{2}$, while $\mathcal{G}_2:\{2 \rightarrow 4\} \rightarrow \mathcal{G}_3:\{4 \rightarrow 5\}$ is a temporal path spanning snapshots $\mathcal{G}_2$ and $\mathcal{G}_3$, contributing weight $-T_2 T_3$. The entry is thus the sum of the two path products: $-\frac{T_2^2}{2} - T_2 T_3$.

The following theorem establishes the main result of this section based on the above graph-theoretic results.

\begin{theorem}
    Consider a temporally switching network with switching sequence $\{\mathcal{G}_1 \rightarrow \mathcal{G}_2 \cdots \rightarrow \mathcal{G}_N\}$ defined over the interval $[t_0,t_N=t_f)$. Let $C$ denote the controllability matrix associated with the temporal network. If, for every node, there exists a path from the leader node in at least one snapshot of the temporal network such that the sign of all such path products is the same, then the network is herdable.
\end{theorem}

        \begin{proof}The predefined switching sequence is $\{\mathcal{G}_1 \rightarrow \mathcal{G}_2 \cdots \rightarrow \mathcal{G}_N\}$ therefore the corresponding system pairs  are $\{(\mathcal{A}_1,\mathcal{B}) \rightarrow (\mathcal{A}_2,\mathcal{B}) \cdots \rightarrow (\mathcal{A}_N,\mathcal{B})\}$. Let $\mathcal{C}(t_f)$ denote the associated controllability matrix defined in \eqref{C-temp}. Let $[pn+1,\, pn+n]$ denote the block of columns of the controllability matrix corresponding to the $(p+1)^{\text{th}}$ snapshot.  For the $(p+1)^{th}$ snapshot, the corresponding block in the controllability matrix \eqref{C-temp} occupies the columns $[pn+1,pn+n]$. The expression for this block is 
\[ \hspace{-1.5mm}
\left(\prod_{k=p+1}^{N} e^{\mathcal{A}_{N-k+(p+1)} T_{N-k+(p+1)}}\right) C_p = 
e^{\mathcal{A}_N T_N} \cdot \cdot e^{\mathcal{A}_{p+1} T_{p+1}} C_p .
\]
Recall that the controllability matrix of the system $\mathcal{A}_p,\mathcal{B}$ is \(
C_p = [\, \mathcal{B} \;\; \mathcal{A}_p \mathcal{B} \;\; \cdots \;\; \mathcal{A}_p^{n-1}\mathcal{B} \,].
\) Therefore, the first column of this block, denoted by $\mathcal{C}(t_f)_{(:,pn+1)}$, is 
\[
C_{(:,pn+1)} =
e^{\mathcal{A}_N T_N}.. e^{\mathcal{A}_{p+1} T_{p+1}} \mathcal{B}
=
\left(\prod_{k=p+1}^{N} e^{\mathcal{A}_k T_k}\right) \mathcal{B} .
\]
Using Taylor series expansion of the matrix exponential \(e^{\mathcal{A}_k T_k}
=
\sum_{m=0}^{\infty} \frac{\mathcal{A}_k^{\,m} T_k^{\,m}}{m!},~~ k \in \{p+1,\dots,N\},\) the above expression becomes a sum of terms involving products of powers of adjacency matrices multiplies by $\mathcal{B}$. Each term corresponds to a sequence of matrix product form 
\[
\mathcal{A}_N^{m_N}\mathcal{A}_{N-1}^{m_{N-1}}\cdots
\mathcal{A}_{p+1}^{m_{p+1}} \mathcal{B}.
\]

We know that $[\mathcal{A}^k_{(i,j)} ]\ne 0$, then there exists a walk of length $k$ from node $i$ to node $j$. From Remark \ref{Remark-path} we know if $[\mathcal{A}_n^k \mathcal{A}_m^l]_{(i,j)} \neq 0$, then there exists a temporal walk from node $i$ to node $j$ of length $k+l$. Similarly, if $[\mathcal{A}^k\mathcal{B}]_i \ne 0$ then there exists a path from the leader node to node $i$ of length $k$. Hence, each nonzero entry in the vector 
\(\left(\prod_{k=p+1}^{N} e^{\mathcal{A}_k T_k}\right)\cdot \mathcal{B} \) represents the contribution of at least one temporal walk from the leader node in all snapshots. If all such temporal walks correspond to each node has same sign, then the resulting vector is unisigned. Consequently, by Proposition \ref{prop:uni}, the temporal network is herdable.\end{proof}

\begin{proposition}
    Consider a temporally switching network with a switching sequence $\{\mathcal{G}_1 \rightarrow \mathcal{G}_2 \cdots \rightarrow \mathcal{G}_N\}$ defined over the interval $[T_0,T_N=t_f)$. If all the walks associated with nodes that are reachable from the leader through walks of length $k-1$ in the snapshot $p+1$,  as well as the nodes that are reachable through temporal walks originating from nodes that are at distance $k-1$ from the leader, have the same path product sign, then the system is herdable. 
\end{proposition}
     
\begin{proof}
    The proof is a direct consequence of Remark~\ref{lem1} and Proposition~\ref{prop:uni}.
\end{proof}
\vspace{-2.5mm}
In the above results, herdability is analyzed for a given realization. In the following section, we examine how the magnitudes of edge weights influence herdability and subsequently study $\mathcal{SS}$ herdability.

\section{A Sufficient Condition for $\mathcal{SS}$ Herdability of Temporally Switching Networks}\label{s5}

It is well known that a directed graph is structurally controllable if all the nodes are spanned by a cactus, which is a disjoint union of stems and buds \cite{lin1974structural}. In our previous work, we derived a graph theoretical condition for $\mathcal{SS}$ herdability of an arbitrary static digraph. Building on the same, in this section, we investigate a similar graph theoretic condition whose existence can guarantee $\mathcal{SS}$ herdability in temporal networks. To motivate the analysis of $\mathcal{SS}$ herdability consider the following example.
\begin{example}
   Consider a temporally switching digraph as shown in the Fig.\ref{ega}.\vspace{-3mm} 
  \begin{figure}[ht]

\begin{subfigure}{0.16\textwidth}
\centering
\tikzset{every picture/.style={scale=0.65pt}} 

\begin{tikzpicture}[x=0.75pt,y=0.75pt,yscale=-1,xscale=1]

\draw  [fill={rgb, 255:red, 0; green, 147; blue, 0 }  ,fill opacity=1 ][line width=0.75]  (128.29,95.49) .. controls (128.29,91.9) and (131.23,88.99) .. (134.84,88.99) .. controls (138.46,88.99) and (141.4,91.9) .. (141.4,95.49) .. controls (141.4,99.08) and (138.46,101.99) .. (134.84,101.99) .. controls (131.23,101.99) and (128.29,99.08) .. (128.29,95.49) -- cycle ;
\draw  [fill={rgb, 255:red, 28; green, 20; blue, 199 }  ,fill opacity=1 ][line width=0.75]  (84.29,134.49) .. controls (84.29,130.9) and (87.23,127.99) .. (90.84,127.99) .. controls (94.46,127.99) and (97.4,130.9) .. (97.4,134.49) .. controls (97.4,138.08) and (94.46,140.99) .. (90.84,140.99) .. controls (87.23,140.99) and (84.29,138.08) .. (84.29,134.49) -- cycle ;
\draw  [fill={rgb, 255:red, 28; green, 20; blue, 199 }  ,fill opacity=1 ][line width=0.75]  (85.29,187.49) .. controls (85.29,183.9) and (88.23,180.99) .. (91.84,180.99) .. controls (95.46,180.99) and (98.4,183.9) .. (98.4,187.49) .. controls (98.4,191.08) and (95.46,193.99) .. (91.84,193.99) .. controls (88.23,193.99) and (85.29,191.08) .. (85.29,187.49) -- cycle ;
\draw  [fill={rgb, 255:red, 28; green, 20; blue, 199 }  ,fill opacity=1 ][line width=0.75]  (174.29,133.49) .. controls (174.29,129.9) and (177.23,126.99) .. (180.84,126.99) .. controls (184.46,126.99) and (187.4,129.9) .. (187.4,133.49) .. controls (187.4,137.08) and (184.46,139.99) .. (180.84,139.99) .. controls (177.23,139.99) and (174.29,137.08) .. (174.29,133.49) -- cycle ;
\draw  [fill={rgb, 255:red, 28; green, 20; blue, 199 }  ,fill opacity=1 ][line width=0.75]  (175.29,186.49) .. controls (175.29,182.9) and (178.23,179.99) .. (181.84,179.99) .. controls (185.46,179.99) and (188.4,182.9) .. (188.4,186.49) .. controls (188.4,190.08) and (185.46,192.99) .. (181.84,192.99) .. controls (178.23,192.99) and (175.29,190.08) .. (175.29,186.49) -- cycle ;
\draw    (128.4,102.99) -- (99.69,127.06) ;
\draw [shift={(97.4,128.99)}, rotate = 320.01] [fill={rgb, 255:red, 0; green, 0; blue, 0 }  ][line width=0.08]  [draw opacity=0] (10.72,-5.15) -- (0,0) -- (10.72,5.15) -- (7.12,0) -- cycle    ;
\draw    (144.2,102.8) -- (171.02,123.37) ;
\draw [shift={(173.4,125.2)}, rotate = 217.49] [fill={rgb, 255:red, 0; green, 0; blue, 0 }  ][line width=0.08]  [draw opacity=0] (10.72,-5.15) -- (0,0) -- (10.72,5.15) -- (7.12,0) -- cycle    ;

\draw (120.6,67.4) node [anchor=north west][inner sep=0.75pt]    {$1$};
\draw (64.2,125.6) node [anchor=north west][inner sep=0.75pt]    {$2$};
\draw (197.4,128.2) node [anchor=north west][inner sep=0.75pt]    {$3$};
\draw (62,179.4) node [anchor=north west][inner sep=0.75pt]    {$4$};
\draw (202,180.4) node [anchor=north west][inner sep=0.75pt]    {$5$};
\draw (80.8,94.52) node [anchor=north west][inner sep=0.75pt]  [font=\footnotesize]  {$a^1_{21}$};
\draw (165,93.72) node [anchor=north west][inner sep=0.75pt]  [font=\footnotesize]  {$a^1_{31}$};

\end{tikzpicture}

\caption{$\mathcal{G}_1(\mathcal{A}_1,\mathcal{B})$}
\end{subfigure}
\centering
\begin{subfigure}{0.15\textwidth}
\centering

\tikzset{every picture/.style={scale=0.65pt}} 

\begin{tikzpicture}[x=0.75pt,y=0.75pt,yscale=-1,xscale=1]

\draw  [fill={rgb, 255:red, 0; green, 147; blue, 0 }  ,fill opacity=1 ][line width=0.75]  (348.29,95.49) .. controls (348.29,91.9) and (351.23,88.99) .. (354.84,88.99) .. controls (358.46,88.99) and (361.4,91.9) .. (361.4,95.49) .. controls (361.4,99.08) and (358.46,101.99) .. (354.84,101.99) .. controls (351.23,101.99) and (348.29,99.08) .. (348.29,95.49) -- cycle ;
\draw  [fill={rgb, 255:red, 28; green, 20; blue, 199 }  ,fill opacity=1 ][line width=0.75]  (304.29,134.49) .. controls (304.29,130.9) and (307.23,127.99) .. (310.84,127.99) .. controls (314.46,127.99) and (317.4,130.9) .. (317.4,134.49) .. controls (317.4,138.08) and (314.46,140.99) .. (310.84,140.99) .. controls (307.23,140.99) and (304.29,138.08) .. (304.29,134.49) -- cycle ;
\draw  [fill={rgb, 255:red, 28; green, 20; blue, 199 }  ,fill opacity=1 ][line width=0.75]  (305.29,187.49) .. controls (305.29,183.9) and (308.23,180.99) .. (311.84,180.99) .. controls (315.46,180.99) and (318.4,183.9) .. (318.4,187.49) .. controls (318.4,191.08) and (315.46,193.99) .. (311.84,193.99) .. controls (308.23,193.99) and (305.29,191.08) .. (305.29,187.49) -- cycle ;
\draw  [fill={rgb, 255:red, 28; green, 20; blue, 199 }  ,fill opacity=1 ][line width=0.75]  (394.29,133.49) .. controls (394.29,129.9) and (397.23,126.99) .. (400.84,126.99) .. controls (404.46,126.99) and (407.4,129.9) .. (407.4,133.49) .. controls (407.4,137.08) and (404.46,139.99) .. (400.84,139.99) .. controls (397.23,139.99) and (394.29,137.08) .. (394.29,133.49) -- cycle ;
\draw  [fill={rgb, 255:red, 28; green, 20; blue, 199 }  ,fill opacity=1 ][line width=0.75]  (398.29,186.49) .. controls (398.29,182.9) and (401.23,179.99) .. (404.84,179.99) .. controls (408.46,179.99) and (411.4,182.9) .. (411.4,186.49) .. controls (411.4,190.08) and (408.46,192.99) .. (404.84,192.99) .. controls (401.23,192.99) and (398.29,190.08) .. (398.29,186.49) -- cycle ;
\draw    (311.4,146.99) -- (310.69,170.98) ;
\draw [shift={(310.6,173.98)}, rotate = 271.69] [fill={rgb, 255:red, 0; green, 0; blue, 0 }  ][line width=0.08]  [draw opacity=0] (10.72,-5.15) -- (0,0) -- (10.72,5.15) -- (7.12,0) -- cycle    ;
\draw    (395,144.4) .. controls (380.1,177.33) and (369.82,181.44) .. (325.36,184.97) ;
\draw [shift={(322.6,185.18)}, rotate = 355.62] [fill={rgb, 255:red, 0; green, 0; blue, 0 }  ][line width=0.08]  [draw opacity=0] (10.72,-5.15) -- (0,0) -- (10.72,5.15) -- (7.12,0) -- cycle    ;

\draw (335.6,70.4) node [anchor=north west][inner sep=0.75pt]    {$1$};
\draw (284.2,125.6) node [anchor=north west][inner sep=0.75pt]    {$2$};
\draw (417.4,128.2) node [anchor=north west][inner sep=0.75pt]    {$3$};
\draw (282,179.4) node [anchor=north west][inner sep=0.75pt]    {$4$};
\draw (425,180.4) node [anchor=north west][inner sep=0.75pt]    {$5$};
\draw (278,146.92) node [anchor=north west][inner sep=0.75pt]  [font=\footnotesize]  {$a^2_{42}$};
\draw (330.6,153.72) node [anchor=north west][inner sep=0.75pt]  [font=\footnotesize]  {$-a^2_{43}$};

\end{tikzpicture}

\caption{$\mathcal{G}_2(\mathcal{A}_2,\mathcal{B})$}
\end{subfigure}
\begin{subfigure}{0.16\textwidth}
\centering

\tikzset{every picture/.style={scale=0.65pt}} 

\begin{tikzpicture}[x=0.75pt,y=0.75pt,yscale=-1,xscale=1]

\draw  [fill={rgb, 255:red, 0; green, 147; blue, 0 }  ,fill opacity=1 ][line width=0.75]  (547.29,96.49) .. controls (547.29,92.9) and (550.23,89.99) .. (553.84,89.99) .. controls (557.46,89.99) and (560.4,92.9) .. (560.4,96.49) .. controls (560.4,100.08) and (557.46,102.99) .. (553.84,102.99) .. controls (550.23,102.99) and (547.29,100.08) .. (547.29,96.49) -- cycle ;
\draw  [fill={rgb, 255:red, 28; green, 20; blue, 199 }  ,fill opacity=1 ][line width=0.75]  (503.29,135.49) .. controls (503.29,131.9) and (506.23,128.99) .. (509.84,128.99) .. controls (513.46,128.99) and (516.4,131.9) .. (516.4,135.49) .. controls (516.4,139.08) and (513.46,141.99) .. (509.84,141.99) .. controls (506.23,141.99) and (503.29,139.08) .. (503.29,135.49) -- cycle ;
\draw  [fill={rgb, 255:red, 28; green, 20; blue, 199 }  ,fill opacity=1 ][line width=0.75]  (504.29,188.49) .. controls (504.29,184.9) and (507.23,181.99) .. (510.84,181.99) .. controls (514.46,181.99) and (517.4,184.9) .. (517.4,188.49) .. controls (517.4,192.08) and (514.46,194.99) .. (510.84,194.99) .. controls (507.23,194.99) and (504.29,192.08) .. (504.29,188.49) -- cycle ;
\draw  [fill={rgb, 255:red, 28; green, 20; blue, 199 }  ,fill opacity=1 ][line width=0.75]  (593.29,134.49) .. controls (593.29,130.9) and (596.23,127.99) .. (599.84,127.99) .. controls (603.46,127.99) and (606.4,130.9) .. (606.4,134.49) .. controls (606.4,138.08) and (603.46,140.99) .. (599.84,140.99) .. controls (596.23,140.99) and (593.29,138.08) .. (593.29,134.49) -- cycle ;
\draw  [fill={rgb, 255:red, 28; green, 20; blue, 199 }  ,fill opacity=1 ][line width=0.75]  (597.29,187.49) .. controls (597.29,183.9) and (600.23,180.99) .. (603.84,180.99) .. controls (607.46,180.99) and (610.4,183.9) .. (610.4,187.49) .. controls (610.4,191.08) and (607.46,193.99) .. (603.84,193.99) .. controls (600.23,193.99) and (597.29,191.08) .. (597.29,187.49) -- cycle ;
\draw    (525.4,134.99) -- (590.4,134.99) ;
\draw [shift={(522.4,134.99)}, rotate = 0] [fill={rgb, 255:red, 0; green, 0; blue, 0 }  ][line width=0.08]  [draw opacity=0] (10.72,-5.15) -- (0,0) -- (10.72,5.15) -- (7.12,0) -- cycle    ;
\draw    (523.4,186.99) -- (588.4,186.99) ;
\draw [shift={(591.4,186.99)}, rotate = 180] [fill={rgb, 255:red, 0; green, 0; blue, 0 }  ][line width=0.08]  [draw opacity=0] (10.72,-5.15) -- (0,0) -- (10.72,5.15) -- (7.12,0) -- cycle    ;

\draw (536.6,70.4) node [anchor=north west][inner sep=0.75pt]    {$1$};
\draw (483.2,126.6) node [anchor=north west][inner sep=0.75pt]    {$2$};
\draw (616.4,129.2) node [anchor=north west][inner sep=0.75pt]    {$3$};
\draw (481,180.4) node [anchor=north west][inner sep=0.75pt]    {$4$};
\draw (624,181.4) node [anchor=north west][inner sep=0.75pt]    {$5$};
\draw (536,111.92) node [anchor=north west][inner sep=0.75pt]  [font=\footnotesize]  {$a^3_{23}$};
\draw (522,162.92) node [anchor=north west][inner sep=0.75pt]  [font=\footnotesize]  {$-a^3
_{54}$};

\end{tikzpicture}

\caption{$\mathcal{G}_3(\mathcal{A}_3,\mathcal{B})$}
\end{subfigure}
\caption{Temporally switching network  $\{\mathcal{A}_1 \rightarrow \mathcal{A}_2 \rightarrow \mathcal{A}_3\}$}
\label{ega}
\end{figure}
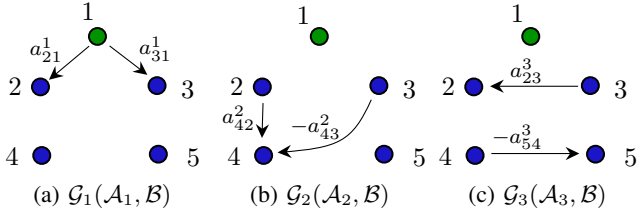   
\end{example} \vspace{-2.5mm}
The controllability matrix associated with the temporal network after $N=3^{rd}$ snapshot is given by:
\[
C_3 =
\left[
\begin{array}{c|c|c}
\begin{array}{cc}
1 & 0 \\
0 & T_3 a^1_{31} a^3_{23} + a^1_{21} \\
0 & a^1_{31} \\
0 & T_2(a^2_{42}-a^2_{43}) \\
0 & T_2(a^2_{42}-a^2_{43}) T_3 a^3_{54}
\end{array}
&
\begin{array}{cc}
1 & 0 \\
0 & \cdots \\
0 & \cdots \\
0 & \cdots \\
0 & \cdots
\end{array}
&
\begin{array}{cc}
1 & 0 \\
0 & \cdots \\
0 & \cdots \\
0 & \cdots \\
0 & \cdots
\end{array}
\end{array}
\right]
\]
 where $T_i=[t_{i-1},t_i)$. Based on our assumption that the magnitude of edge weights $a^k_{ij} \in \mathbb{R}^+$, it is evident that the temporal network is herdable only when  $a^2_{42}>a^2_{43}$. When $a^2_{42} = a^2_{43}$, rows $4$ and $5$ vanish identically, and the network fails to be herdable; herdability instead holds whenever $a^2_{42} > a^2_{43}$, under the stated assumption. Hence the given temporal network is $\mathcal{SS}$ herdable. We now formally introduce the definition of $\mathcal{SS}$ herdability for temporal networks.

\begin{definition}
 Given the zero-nonzero patterns of $(\mathcal{A}_i, \mathcal{B})_{i=1}^N$, system~\eqref{sys1} is said to be structurally sign $\mathcal{SS}$ herdable if there exists at least one realization of $(\mathcal{A}_i, \mathcal{B})_{i=1}^N$, consistent with these patterns that is herdable over some temporal sequence.   
\end{definition}

\begin{definition}
Let $(A_i, B_i)_{i=1}^N$ be a sequence of matrices specified by their zero-nonzero patterns.  system~\eqref{sys1} is said to be structurally sign $\mathcal{SS}$ herdable if there exists a sequence of matrices $(\tilde{A}_i, \tilde{B}_i)_{i=1}^N$ such that:
\begin{enumerate}
\item For each $i \in \{1, \dots, N\}$, the matrices $\tilde{A}_i$ and $\tilde{B}_i$ have the same zero-nonzero and sign patterns as $\mathcal{A}_i$ and $\mathcal{B}_i$, respectively; that is, for every nonzero entry, $\operatorname{sign}(\tilde{a}_{ij}^{(k)}) = \operatorname{sign}(a_{ij}^{(k)})$ and $\operatorname{sign}(\tilde{b}_{i}^{(k)}) = \operatorname{sign}(b_{i}^{(k)})$, while their magnitudes may vary.
    \item The resulting system evolving according to the temporal sequence $(\tilde{A}_i, \tilde{B}_i)_{i=1}^N$, is herdable.
\end{enumerate}
\end{definition}

 \subsection{Union Multigraph $\mathcal{G}_\mathfrak{U}(\mathcal{A,B})$}
 Let $\mathcal{G}_i(\mathcal{A},\mathcal{B})$, $i \in \{1,\dots,N\}$,, denote the digraph associated with $i^{th}$
subsystem  $(\mathcal{A}_i,\mathcal{B})$. The union multigraph is defined as \vspace{-1mm}
\[
\mathcal{G}_{\mathfrak{U}}(\mathcal{A},\mathcal{B}) = \biguplus_{i=1}^N \mathcal{G}_i(\mathcal{A}_i,\mathcal{B}),
\]
where $\mathcal{G}_{\mathfrak{U}}(\mathcal{A}_i,\mathcal{B})$ contains the union of the node sets and multiset union of edge sets of all $\mathcal{G}_i$, preserving edge weights and multiplicities. For example, consider the temporal network shown in Fig.~\ref{mul-full}. The corresponding union multigraph $\mathcal{G}_{\mathfrak{U}}(\mathcal{A},\mathcal{B})$ is illustrated in Fig.~\ref{mul-full}(\subref{mul-tog}).

  \begin{figure}[ht]

\begin{subfigure}{0.21\textwidth}
\centering
\tikzset{every picture/.style={scale=0.75pt}} 

\begin{tikzpicture}[x=0.75pt,y=0.75pt,yscale=-1,xscale=1]

\draw [color={rgb, 255:red, 0; green, 0; blue, 0 }  ,draw opacity=1 ]   (145.46,183.65) -- (144.76,207.65) ;
\draw [shift={(144.67,210.65)}, rotate = 271.69] [fill={rgb, 255:red, 0; green, 0; blue, 0 }  ,fill opacity=1 ][line width=0.08]  [draw opacity=0] (10.72,-5.15) -- (0,0) -- (10.72,5.15) -- (7.12,0) -- cycle    ;
\draw  [fill={rgb, 255:red, 0; green, 147; blue, 0 }  ,fill opacity=1 ][line width=0.75]  (181.33,132.49) .. controls (181.33,128.9) and (184.26,125.99) .. (187.88,125.99) .. controls (191.5,125.99) and (194.43,128.9) .. (194.43,132.49) .. controls (194.43,136.08) and (191.5,138.99) .. (187.88,138.99) .. controls (184.26,138.99) and (181.33,136.08) .. (181.33,132.49) -- cycle ;
\draw  [fill={rgb, 255:red, 28; green, 20; blue, 199 }  ,fill opacity=1 ][line width=0.75]  (137.33,171.49) .. controls (137.33,167.9) and (140.26,164.99) .. (143.88,164.99) .. controls (147.5,164.99) and (150.43,167.9) .. (150.43,171.49) .. controls (150.43,175.08) and (147.5,177.99) .. (143.88,177.99) .. controls (140.26,177.99) and (137.33,175.08) .. (137.33,171.49) -- cycle ;
\draw  [fill={rgb, 255:red, 28; green, 20; blue, 199 }  ,fill opacity=1 ][line width=0.75]  (138.33,224.49) .. controls (138.33,220.9) and (141.26,217.99) .. (144.88,217.99) .. controls (148.5,217.99) and (151.43,220.9) .. (151.43,224.49) .. controls (151.43,228.08) and (148.5,230.99) .. (144.88,230.99) .. controls (141.26,230.99) and (138.33,228.08) .. (138.33,224.49) -- cycle ;
\draw  [fill={rgb, 255:red, 28; green, 20; blue, 199 }  ,fill opacity=1 ][line width=0.75]  (227.33,170.49) .. controls (227.33,166.9) and (230.26,163.99) .. (233.88,163.99) .. controls (237.5,163.99) and (240.43,166.9) .. (240.43,170.49) .. controls (240.43,174.08) and (237.5,176.99) .. (233.88,176.99) .. controls (230.26,176.99) and (227.33,174.08) .. (227.33,170.49) -- cycle ;
\draw  [fill={rgb, 255:red, 28; green, 20; blue, 199 }  ,fill opacity=1 ][line width=0.75]  (226.66,220.49) .. controls (226.66,216.9) and (229.59,213.99) .. (233.21,213.99) .. controls (236.83,213.99) and (239.76,216.9) .. (239.76,220.49) .. controls (239.76,224.08) and (236.83,226.99) .. (233.21,226.99) .. controls (229.59,226.99) and (226.66,224.08) .. (226.66,220.49) -- cycle ;
\draw [fill={rgb, 255:red, 210; green, 39; blue, 39 }  ,fill opacity=1 ][line width=0.75]    (181.43,139.99) -- (152.73,164.06) ;
\draw [shift={(150.43,165.99)}, rotate = 320.01] [fill={rgb, 255:red, 0; green, 0; blue, 0 }  ][line width=0.08]  [draw opacity=0] (10.72,-5.15) -- (0,0) -- (10.72,5.15) -- (7.12,0) -- cycle    ;
\draw [fill={rgb, 255:red, 210; green, 39; blue, 39 }  ,fill opacity=1 ][line width=0.75]    (197.23,139.8) -- (224.05,160.37) ;
\draw [shift={(226.43,162.2)}, rotate = 217.49] [fill={rgb, 255:red, 0; green, 0; blue, 0 }  ][line width=0.08]  [draw opacity=0] (10.72,-5.15) -- (0,0) -- (10.72,5.15) -- (7.12,0) -- cycle    ;
\draw [color={rgb, 255:red, 0; green, 0; blue, 0 }  ,draw opacity=1 ]   (157.9,180.6) -- (219.02,211.45) ;
\draw [shift={(221.7,212.8)}, rotate = 206.78] [fill={rgb, 255:red, 0; green, 0; blue, 0 }  ,fill opacity=1 ][line width=0.08]  [draw opacity=0] (10.72,-5.15) -- (0,0) -- (10.72,5.15) -- (7.12,0) -- cycle    ;

\draw (166.63,121.4) node [anchor=north west][inner sep=0.75pt]    {$1$};
\draw (118.23,151.27) node [anchor=north west][inner sep=0.75pt]    {$2$};
\draw (245.1,149.2) node [anchor=north west][inner sep=0.75pt]    {$3$};
\draw (119.03,224.07) node [anchor=north west][inner sep=0.75pt]    {$4$};
\draw (249.03,224.4) node [anchor=north west][inner sep=0.75pt]    {$5$};
\draw (142.52,134.92) node [anchor=north west][inner sep=0.75pt]  [font=\scriptsize]  {$a_{21}^{1}$};
\draw (211.81,133.9) node [anchor=north west][inner sep=0.75pt]  [font=\scriptsize]  {$-a_{31}^{1}$};
\draw (179.52,176.59) node [anchor=north west][inner sep=0.75pt]  [font=\scriptsize]  {$-a_{52}^{1}$};
\draw (117.33,184.85) node [anchor=north west][inner sep=0.75pt]  [font=\scriptsize]  {$a_{42}^{1}$};

\end{tikzpicture}
    
\caption{$\mathcal{G}_1(\mathcal{A}_1,\mathcal{B})$}
\label{mul-1}
\end{subfigure}
\centering
\begin{subfigure}{0.21\textwidth}
\centering

\tikzset{every picture/.style={scale=0.75pt}} 

\begin{tikzpicture}[x=0.75pt,y=0.75pt,yscale=-1,xscale=1]

\draw  [fill={rgb, 255:red, 0; green, 147; blue, 0 }  ,fill opacity=1 ][line width=0.75]  (408.29,133.49) .. controls (408.29,129.9) and (411.23,126.99) .. (414.84,126.99) .. controls (418.46,126.99) and (421.4,129.9) .. (421.4,133.49) .. controls (421.4,137.08) and (418.46,139.99) .. (414.84,139.99) .. controls (411.23,139.99) and (408.29,137.08) .. (408.29,133.49) -- cycle ;
\draw  [fill={rgb, 255:red, 28; green, 20; blue, 199 }  ,fill opacity=1 ][line width=0.75]  (364.29,172.49) .. controls (364.29,168.9) and (367.23,165.99) .. (370.84,165.99) .. controls (374.46,165.99) and (377.4,168.9) .. (377.4,172.49) .. controls (377.4,176.08) and (374.46,178.99) .. (370.84,178.99) .. controls (367.23,178.99) and (364.29,176.08) .. (364.29,172.49) -- cycle ;
\draw  [fill={rgb, 255:red, 28; green, 20; blue, 199 }  ,fill opacity=1 ][line width=0.75]  (454.29,171.49) .. controls (454.29,167.9) and (457.23,164.99) .. (460.84,164.99) .. controls (464.46,164.99) and (467.4,167.9) .. (467.4,171.49) .. controls (467.4,175.08) and (464.46,177.99) .. (460.84,177.99) .. controls (457.23,177.99) and (454.29,175.08) .. (454.29,171.49) -- cycle ;
\draw  [fill={rgb, 255:red, 28; green, 20; blue, 199 }  ,fill opacity=1 ][line width=0.75]  (364.29,225.49) .. controls (364.29,221.9) and (367.23,218.99) .. (370.84,218.99) .. controls (374.46,218.99) and (377.4,221.9) .. (377.4,225.49) .. controls (377.4,229.08) and (374.46,231.99) .. (370.84,231.99) .. controls (367.23,231.99) and (364.29,229.08) .. (364.29,225.49) -- cycle ;
\draw  [fill={rgb, 255:red, 28; green, 20; blue, 199 }  ,fill opacity=1 ][line width=0.75]  (454.29,224.49) .. controls (454.29,220.9) and (457.23,217.99) .. (460.84,217.99) .. controls (464.46,217.99) and (467.4,220.9) .. (467.4,224.49) .. controls (467.4,228.08) and (464.46,230.99) .. (460.84,230.99) .. controls (457.23,230.99) and (454.29,228.08) .. (454.29,224.49) -- cycle ;
\draw [color={rgb, 255:red, 145; green, 39; blue, 39 }  ,draw opacity=1 ]   (446.4,172.99) -- (388,173) ;
\draw [shift={(385,173)}, rotate = 359.99] [fill={rgb, 255:red, 145; green, 39; blue, 39 }  ,fill opacity=1 ][line width=0.08]  [draw opacity=0] (10.72,-5.15) -- (0,0) -- (10.72,5.15) -- (7.12,0) -- cycle    ;
\draw [color={rgb, 255:red, 145; green, 39; blue, 39 }  ,draw opacity=1 ]   (446.3,226) -- (385.3,226) ;
\draw [shift={(449.3,226)}, rotate = 180] [fill={rgb, 255:red, 145; green, 39; blue, 39 }  ,fill opacity=1 ][line width=0.08]  [draw opacity=0] (10.72,-5.15) -- (0,0) -- (10.72,5.15) -- (7.12,0) -- cycle    ;

\draw (389.6,126.4) node [anchor=north west][inner sep=0.75pt]    {$1$};
\draw (344.2,156.6) node [anchor=north west][inner sep=0.75pt]    {$2$};
\draw (473.4,153.2) node [anchor=north west][inner sep=0.75pt]    {$3$};
\draw (344,226.4) node [anchor=north west][inner sep=0.75pt]    {$4$};
\draw (473,223.4) node [anchor=north west][inner sep=0.75pt]    {$5$};
\draw (407.4,152.32) node [anchor=north west][inner sep=0.75pt]  [font=\scriptsize,color={rgb, 255:red, 145; green, 39; blue, 39 }  ,opacity=1 ]  {$a_{23}^{2}$};
\draw (400.33,202.79) node [anchor=north west][inner sep=0.75pt]  [font=\scriptsize,color={rgb, 255:red, 145; green, 39; blue, 39 }  ,opacity=1 ]  {$-a_{54}^{2}$};

\end{tikzpicture}

\caption{$\mathcal{G}_2(\mathcal{A}_2,\mathcal{B})$}
\label{mul-2}
\end{subfigure}

\vspace{2mm}

\begin{subfigure}{0.21\textwidth}
\centering

\tikzset{every picture/.style={scale=0.75pt}} 

\begin{tikzpicture}[x=0.75pt,y=0.75pt,yscale=-1,xscale=1]

\draw  [fill={rgb, 255:red, 0; green, 147; blue, 0 }  ,fill opacity=1 ][line width=0.75]  (282.96,100.49) .. controls (282.96,96.9) and (285.89,93.99) .. (289.51,93.99) .. controls (293.13,93.99) and (296.06,96.9) .. (296.06,100.49) .. controls (296.06,104.08) and (293.13,106.99) .. (289.51,106.99) .. controls (285.89,106.99) and (282.96,104.08) .. (282.96,100.49) -- cycle ;
\draw  [fill={rgb, 255:red, 28; green, 20; blue, 199 }  ,fill opacity=1 ][line width=0.75]  (238.96,139.49) .. controls (238.96,135.9) and (241.89,132.99) .. (245.51,132.99) .. controls (249.13,132.99) and (252.06,135.9) .. (252.06,139.49) .. controls (252.06,143.08) and (249.13,145.99) .. (245.51,145.99) .. controls (241.89,145.99) and (238.96,143.08) .. (238.96,139.49) -- cycle ;
\draw  [fill={rgb, 255:red, 28; green, 20; blue, 199 }  ,fill opacity=1 ][line width=0.75]  (328.96,138.49) .. controls (328.96,134.9) and (331.89,131.99) .. (335.51,131.99) .. controls (339.13,131.99) and (342.06,134.9) .. (342.06,138.49) .. controls (342.06,142.08) and (339.13,144.99) .. (335.51,144.99) .. controls (331.89,144.99) and (328.96,142.08) .. (328.96,138.49) -- cycle ;
\draw  [fill={rgb, 255:red, 0; green, 147; blue, 0 }  ,fill opacity=1 ][line width=0.75]  (282.96,100.49) .. controls (282.96,96.9) and (285.89,93.99) .. (289.51,93.99) .. controls (293.13,93.99) and (296.06,96.9) .. (296.06,100.49) .. controls (296.06,104.08) and (293.13,106.99) .. (289.51,106.99) .. controls (285.89,106.99) and (282.96,104.08) .. (282.96,100.49) -- cycle ;
\draw  [fill={rgb, 255:red, 28; green, 20; blue, 199 }  ,fill opacity=1 ][line width=0.75]  (238.96,139.49) .. controls (238.96,135.9) and (241.89,132.99) .. (245.51,132.99) .. controls (249.13,132.99) and (252.06,135.9) .. (252.06,139.49) .. controls (252.06,143.08) and (249.13,145.99) .. (245.51,145.99) .. controls (241.89,145.99) and (238.96,143.08) .. (238.96,139.49) -- cycle ;
\draw  [fill={rgb, 255:red, 28; green, 20; blue, 199 }  ,fill opacity=1 ][line width=0.75]  (328.96,138.49) .. controls (328.96,134.9) and (331.89,131.99) .. (335.51,131.99) .. controls (339.13,131.99) and (342.06,134.9) .. (342.06,138.49) .. controls (342.06,142.08) and (339.13,144.99) .. (335.51,144.99) .. controls (331.89,144.99) and (328.96,142.08) .. (328.96,138.49) -- cycle ;
\draw  [fill={rgb, 255:red, 28; green, 20; blue, 199 }  ,fill opacity=1 ][line width=0.75]  (239.96,190.49) .. controls (239.96,186.9) and (242.89,183.99) .. (246.51,183.99) .. controls (250.13,183.99) and (253.06,186.9) .. (253.06,190.49) .. controls (253.06,194.08) and (250.13,196.99) .. (246.51,196.99) .. controls (242.89,196.99) and (239.96,194.08) .. (239.96,190.49) -- cycle ;
\draw  [fill={rgb, 255:red, 28; green, 20; blue, 199 }  ,fill opacity=1 ][line width=0.75]  (329.96,189.49) .. controls (329.96,185.9) and (332.89,182.99) .. (336.51,182.99) .. controls (340.13,182.99) and (343.06,185.9) .. (343.06,189.49) .. controls (343.06,193.08) and (340.13,195.99) .. (336.51,195.99) .. controls (332.89,195.99) and (329.96,193.08) .. (329.96,189.49) -- cycle ;
\draw [color={rgb, 255:red, 18; green, 141; blue, 113 }  ,draw opacity=1 ]   (245.76,151.99) -- (245.06,175.98) ;
\draw [shift={(244.97,178.98)}, rotate = 271.69] [fill={rgb, 255:red, 18; green, 141; blue, 113 }  ,fill opacity=1 ][line width=0.08]  [draw opacity=0] (10.72,-5.15) -- (0,0) -- (10.72,5.15) -- (7.12,0) -- cycle    ;
\draw [color={rgb, 255:red, 18; green, 141; blue, 113 }  ,draw opacity=1 ]   (260.23,150.47) -- (321.36,181.31) ;
\draw [shift={(324.03,182.67)}, rotate = 206.78] [fill={rgb, 255:red, 18; green, 141; blue, 113 }  ,fill opacity=1 ][line width=0.08]  [draw opacity=0] (10.72,-5.15) -- (0,0) -- (10.72,5.15) -- (7.12,0) -- cycle    ;

\draw (224.87,115.6) node [anchor=north west][inner sep=0.75pt]    {$2$};
\draw (349.07,117.2) node [anchor=north west][inner sep=0.75pt]    {$3$};
\draw (266.27,91.4) node [anchor=north west][inner sep=0.75pt]    {$1$};
\draw (220.67,191.4) node [anchor=north west][inner sep=0.75pt]    {$4$};
\draw (353.67,192.4) node [anchor=north west][inner sep=0.75pt]    {$5$};
\draw (208,152.41) node [anchor=north west][inner sep=0.75pt]  [font=\scriptsize,color={rgb, 255:red, 18; green, 141; blue, 113 }  ,opacity=1 ]  {$a_{42}^{3}$};
\draw (279.65,137.59) node [anchor=north west][inner sep=0.75pt]  [font=\scriptsize,color={rgb, 255:red, 18; green, 141; blue, 113 }  ,opacity=1 ]  {$a_{52}^{3}$};

\end{tikzpicture}

\caption{$\mathcal{G}_3(\mathcal{A}_3,\mathcal{B})$}
\label{mul-3}
\end{subfigure}
\begin{subfigure}{0.21\textwidth}
\centering

\tikzset{every picture/.style={scale=0.9pt}} 

\begin{tikzpicture}[x=0.75pt,y=0.75pt,yscale=-1,xscale=1]

\draw  [draw opacity=0][fill={rgb, 255:red, 241; green, 229; blue, 88 }  ,fill opacity=0.63 ][dash pattern={on 4.5pt off 4.5pt}] (327.11,85.95) .. controls (288.82,64.14) and (297.97,41) .. (316.82,50.81) .. controls (335.68,60.62) and (362.54,77.47) .. (367.68,97.09) .. controls (372.82,116.71) and (321.3,97.35) .. (320.25,103.66) .. controls (319.2,109.98) and (365.58,134.97) .. (371.11,149.38) .. controls (376.64,163.79) and (359.08,175.71) .. (350.67,164.54) .. controls (342.25,153.38) and (343.39,134.43) .. (327.39,122.43) .. controls (311.39,110.43) and (280.25,107.95) .. (271,105.63) .. controls (261.75,103.3) and (283.68,154.52) .. (276.49,162.59) .. controls (269.31,170.66) and (260.7,159.9) .. (259.39,156.81) .. controls (258.09,153.71) and (249.97,119.66) .. (255.68,94.81) .. controls (261.39,69.95) and (365.39,107.76) .. (327.11,85.95) -- cycle ;
\draw [color={rgb, 255:red, 18; green, 141; blue, 113 }  ,draw opacity=1 ]   (264.54,107.95) -- (266.22,132.64) ;
\draw [shift={(266.43,135.63)}, rotate = 266.09] [fill={rgb, 255:red, 18; green, 141; blue, 113 }  ,fill opacity=1 ][line width=0.08]  [draw opacity=0] (10.72,-5.15) -- (0,0) -- (10.72,5.15) -- (7.12,0) -- cycle    ;
\draw [color={rgb, 255:red, 145; green, 39; blue, 39 }  ,draw opacity=1 ]   (341.43,97.11) -- (283.03,97.13) ;
\draw [shift={(280.03,97.13)}, rotate = 359.99] [fill={rgb, 255:red, 145; green, 39; blue, 39 }  ,fill opacity=1 ][line width=0.08]  [draw opacity=0] (10.72,-5.15) -- (0,0) -- (10.72,5.15) -- (7.12,0) -- cycle    ;
\draw [color={rgb, 255:red, 145; green, 39; blue, 39 }  ,draw opacity=1 ]   (341.33,150.13) -- (280.33,150.13) ;
\draw [shift={(344.33,150.13)}, rotate = 180] [fill={rgb, 255:red, 145; green, 39; blue, 39 }  ,fill opacity=1 ][line width=0.08]  [draw opacity=0] (10.72,-5.15) -- (0,0) -- (10.72,5.15) -- (7.12,0) -- cycle    ;
\draw  [fill={rgb, 255:red, 0; green, 147; blue, 0 }  ,fill opacity=1 ][line width=0.75]  (303.66,58.61) .. controls (303.66,55.02) and (306.59,52.11) .. (310.21,52.11) .. controls (313.83,52.11) and (316.76,55.02) .. (316.76,58.61) .. controls (316.76,62.2) and (313.83,65.11) .. (310.21,65.11) .. controls (306.59,65.11) and (303.66,62.2) .. (303.66,58.61) -- cycle ;
\draw  [fill={rgb, 255:red, 28; green, 20; blue, 199 }  ,fill opacity=1 ][line width=0.75]  (259.66,97.61) .. controls (259.66,94.02) and (262.59,91.11) .. (266.21,91.11) .. controls (269.83,91.11) and (272.76,94.02) .. (272.76,97.61) .. controls (272.76,101.2) and (269.83,104.11) .. (266.21,104.11) .. controls (262.59,104.11) and (259.66,101.2) .. (259.66,97.61) -- cycle ;
\draw  [fill={rgb, 255:red, 28; green, 20; blue, 199 }  ,fill opacity=1 ][line width=0.75]  (260.66,150.61) .. controls (260.66,147.02) and (263.59,144.11) .. (267.21,144.11) .. controls (270.83,144.11) and (273.76,147.02) .. (273.76,150.61) .. controls (273.76,154.2) and (270.83,157.11) .. (267.21,157.11) .. controls (263.59,157.11) and (260.66,154.2) .. (260.66,150.61) -- cycle ;
\draw  [fill={rgb, 255:red, 28; green, 20; blue, 199 }  ,fill opacity=1 ][line width=0.75]  (349.66,96.61) .. controls (349.66,93.02) and (352.59,90.11) .. (356.21,90.11) .. controls (359.83,90.11) and (362.76,93.02) .. (362.76,96.61) .. controls (362.76,100.2) and (359.83,103.11) .. (356.21,103.11) .. controls (352.59,103.11) and (349.66,100.2) .. (349.66,96.61) -- cycle ;
\draw  [fill={rgb, 255:red, 28; green, 20; blue, 199 }  ,fill opacity=1 ][line width=0.75]  (350.66,149.61) .. controls (350.66,146.02) and (353.59,143.11) .. (357.21,143.11) .. controls (360.83,143.11) and (363.76,146.02) .. (363.76,149.61) .. controls (363.76,153.2) and (360.83,156.11) .. (357.21,156.11) .. controls (353.59,156.11) and (350.66,153.2) .. (350.66,149.61) -- cycle ;
\draw    (303.76,66.11) -- (275.06,90.19) ;
\draw [shift={(272.76,92.11)}, rotate = 320.01] [fill={rgb, 255:red, 0; green, 0; blue, 0 }  ][line width=0.08]  [draw opacity=0] (10.72,-5.15) -- (0,0) -- (10.72,5.15) -- (7.12,0) -- cycle    ;
\draw    (319.57,65.93) -- (346.39,86.5) ;
\draw [shift={(348.77,88.33)}, rotate = 217.49] [fill={rgb, 255:red, 0; green, 0; blue, 0 }  ][line width=0.08]  [draw opacity=0] (10.72,-5.15) -- (0,0) -- (10.72,5.15) -- (7.12,0) -- cycle    ;
\draw    (253.11,102.52) .. controls (240.64,112.5) and (232.47,128.12) .. (252.85,146.78) ;
\draw [shift={(254.82,148.52)}, rotate = 220.36] [fill={rgb, 255:red, 0; green, 0; blue, 0 }  ][line width=0.08]  [draw opacity=0] (10.72,-5.15) -- (0,0) -- (10.72,5.15) -- (7.12,0) -- cycle    ;
\draw [color={rgb, 255:red, 18; green, 141; blue, 113 }  ,draw opacity=1 ]   (274.53,102.99) .. controls (291.14,111.81) and (310.34,89.67) .. (349.85,135.8) ;
\draw [shift={(351.67,137.96)}, rotate = 230.24] [fill={rgb, 255:red, 18; green, 141; blue, 113 }  ,fill opacity=1 ][line width=0.08]  [draw opacity=0] (10.72,-5.15) -- (0,0) -- (10.72,5.15) -- (7.12,0) -- cycle    ;
\draw    (273.11,108.81) .. controls (310.87,125.88) and (283.38,135.5) .. (341.93,141.68) ;
\draw [shift={(344.67,141.96)}, rotate = 185.67] [fill={rgb, 255:red, 0; green, 0; blue, 0 }  ][line width=0.08]  [draw opacity=0] (10.72,-5.15) -- (0,0) -- (10.72,5.15) -- (7.12,0) -- cycle    ;

\draw (271.59,116.87) node [anchor=north west][inner sep=0.75pt]  [font=\scriptsize,color={rgb, 255:red, 18; green, 141; blue, 113 }  ,opacity=1 ]  {$a_{42}^{3}$};
\draw (346.05,110.87) node [anchor=north west][inner sep=0.75pt]  [font=\scriptsize,color={rgb, 255:red, 18; green, 141; blue, 113 }  ,opacity=1 ]  {$a_{52}^{3}$};
\draw (297.83,74.46) node [anchor=north west][inner sep=0.75pt]  [font=\scriptsize,color={rgb, 255:red, 145; green, 39; blue, 39 }  ,opacity=1 ]  {$a_{23}^{2}$};
\draw (297.6,153.54) node [anchor=north west][inner sep=0.75pt]  [font=\scriptsize,color={rgb, 255:red, 145; green, 39; blue, 39 }  ,opacity=1 ]  {$-a_{54}^{2}$};
\draw (289.63,42.86) node [anchor=north west][inner sep=0.75pt]    {$1$};
\draw (243.57,77.06) node [anchor=north west][inner sep=0.75pt]    {$2$};
\draw (363.1,70.99) node [anchor=north west][inner sep=0.75pt]    {$3$};
\draw (247.7,157.53) node [anchor=north west][inner sep=0.75pt]    {$4$};
\draw (369.37,156.53) node [anchor=north west][inner sep=0.75pt]    {$5$};
\draw (270.86,61.64) node [anchor=north west][inner sep=0.75pt]  [font=\scriptsize]  {$a_{21}^{1}$};
\draw (335.25,55.13) node [anchor=north west][inner sep=0.75pt]  [font=\scriptsize]  {$-a_{31}^{1}$};
\draw (218.61,115.65) node [anchor=north west][inner sep=0.75pt]  [font=\scriptsize]  {$a_{42}^{1}$};
\draw (300.33,120.94) node [anchor=north west][inner sep=0.75pt]  [font=\scriptsize]  {$-a_{52}^{1}$};

\end{tikzpicture}

\caption{$\mathcal{G}_\mathfrak{U}(\mathcal{A}_i,\mathcal{B})$}
\label{mul-tog}
\end{subfigure}

\caption{The union graph $\mathcal{G}_{\mathfrak{U}}(\mathcal{A},\mathcal{B})$ (Fig.~\subref{mul-tog}) is the multiset union of $\mathcal{G}_1$-$\mathcal{G}_3$ (Figs.~\subref{mul-1}-\subref{mul-3}), preserving node sets, edge weights, and multiplicities.}

\label{mul-full}
\end{figure}
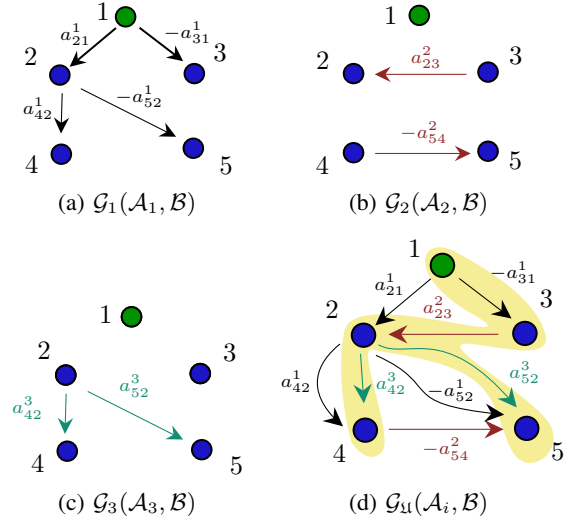   

\vspace{-3mm}

\subsection{$\pi$-Graph and Conditions for $\mathcal{SS}$ Herdability}
\begin{definition}[$\pi$-graph]\label{pi-def}
Consider the temporal network $(\mathcal{A}_i,\mathcal{B})_{i=1}^N$, where each pair $(\mathcal{A}_i,\mathcal{B})$ represents the system at snapshot $i$, with associated digraphs $\mathcal{G}_i(\mathcal{A}_i,\mathcal{B})$. Let
\(\mathcal{G}_{\mathfrak{U}}(\mathcal{A},\mathcal{B})\)
denote the corresponding union multigraph. A $\pi$-graph is a spanning subgraph of $\mathcal{G}_{\mathfrak{U}}(\mathcal{A},\mathcal{B})$ such that:
\begin{itemize}
    \item it is input-connected, and
    \item \(\exists\,\sigma\in\{+1,-1\}:
\operatorname{sgn}(P)=\sigma,
\forall v\in\mathcal{V}_{\pi},\ 
\forall P\in\mathcal{P}_{\pi}(1,v).
\), where \(\mathcal V_\pi\) is the set of nodes in the \(\pi\)- graph and \(\mathcal P_\pi(1,v)\) is the set of all walks (or paths ) from the leader node 1 to \(v\).
\end{itemize}
\end{definition}
\vspace{-2mm}
\begin{definition}[Temporal $\pi$-graph]\label{Temp-pi-def}
A temporal $\pi$-graph, denoted by $\pi_p$, is a $\pi$-graph in which every node is reachable from the leader node via temporal walks that preserve a consistent path sign. The $\pi_p$-graph evolves over the sequence of time intervals $\{T_i\}_{i=1}^N$.
\end{definition}
For example, in Fig.~\ref{mul-tog}, the $\pi_p$-graph is highlighted, where all paths from the leader to every node have the same path sign. This property leads to the following result.
\begin{proposition}\label{prop-pi}
 
Let $\pi_p$- graph be a temporally evolving graph defined by Definition \ref{Temp-pi-def}, that spans all nodes. Let $\mathcal{C}_{\pi} \in \mathbb{R}^{n\times nN}$, denote the associated controllability matrix. Then, $\mathcal{C}_{\pi}$ admits a strictly positive image, i.e.,\(\exists\, \bm{v} \in \mathbb{R}^{nN} \quad \text{such that} \quad \mathcal{C}_{\pi}\cdot\bm{v} \in \mathbb{R}^{n}_{>0}.\) Consequently, the system is $\mathcal{SS}$ herdable.
\end{proposition}
\begin{proof}
Since the underlying temporally evolving $\pi$-graph ensures that all temporal walks from the leader node to any node have the same path sign, it follows from Remark~\ref{lem1} that all entries of the associated controllability matrix $\mathcal{C}_{\pi}$ share a common sign. Hence, there exists a vector $\bm{v}$ such that $\mathcal{C}_{\pi}\cdot\bm{v} \in \mathbb{R}^n_{>0}$, implying that $\mathcal{C}_{\pi}$ admits a strictly positive image. Therefore, the system is $\mathcal{SS}$ herdable.
\end{proof}
\begin{remark} \label{pi-col}
Since the $\pi_p$-graph temporally span all the nodes over the sequence of snapshots, Remark~\ref{lem1} implies that $\mathcal{C}_{(:,pN+2)}$ is the column corresponding to the $\pi_p$-graph for the $p+1$ snapshot sequence. By definition of the $\pi_p$-graph, all paths contributing to this column corresponds to the followers have the same sign. Hence $\mathcal{C}_{(:,pN+2)}$ is unisigned.
\end{remark}

\begin{theorem}
Consider the temporal network $(\mathcal{A}_i,\mathcal{B})_{i=1}^N$, whose dynamics are governed by \eqref{sys1}, where each pair $(\mathcal{A}_i,\mathcal{B})$ represents the system at snapshot $i$, with associated digraphs $\mathcal{G}_i(\mathcal{A}_i,\mathcal{B})$. Let $
\mathcal{G}_{\mathfrak{U}}(\mathcal{A},\mathcal{B})$ denote the corresponding union multigraph. If $\mathcal{G}_{\mathfrak{U}}(\mathcal{A},\mathcal{B})$ has a $\pi_p$- graph spanning all the nodes across the temporal sequence, then the temporal network \eqref{sys1} is $\mathcal{SS}$ herdable.
\end{theorem}

\begin{proof}
Consider the union multigraph of the temporally evolving network and suppose that it contains a $\pi_p$-graph spanning all nodes. By Remark~\ref{pi-col}, the $\pi_p$-graph contributes to the column $\mathcal{C}_{(:,pN+2)}$ of $\mathcal{C}(t_f)$. For each node $v_j$, let \(\mathcal{C}_{(j,pN+2)}=D_j+R_j\), where $D_j$ denotes the contribution from the $\pi_p$-graph and $R_j$ denotes the aggregate contribution from all remaining paths. Since the edge weights are free parameters in the structural setting, the weights associated with the $\pi_p$-graph can be chosen with sufficiently large magnitude relative to the remaining edge weights such that \(|D_j|>|R_j|, \forall j.\)
By the definition of a $\pi_p$-graph, the contributions $D_j$ are nonzero and have the same sign for all nodes $v_j$. Hence, $\operatorname{sgn}(\mathcal{C}_{(j,pN+2)})=\operatorname{sgn}(D_j)$ for every $j$, making $\mathcal{C}_{(:,pN+2)}$ a unisigned column with nonzero follower entries. Thus, the network is herdable.
\end{proof}

The above theorem guarantees that if there exists a temporally evolving $\pi_p$-graph spanning all nodes in the multigraph, then the network is $\mathcal{SS}$ herdable. Cycles introduce a further complication: the corresponding entries of the controllability matrix become time-varying analytic functions of $\sin t$ and $\cos t$, rather than fixed scalars. Consequently, the columns of the matrix span a function space rather than $\mathbb{R}^n$, and verifying herdability requires identifying a time domain over which a strictly positive linear combination of these functions exists.

\section{Conclusions}\label{s6}
 In this letter, we studied herdability in temporally switching directed networks. Using the relationship between temporal walks and the controllability matrix, we derived necessary and sufficient conditions for herdability and showed how the magnitudes of the edge weight influences the herdability. We then investigated structural sign ($\mathcal{SS}$) herdability and introduced a sufficient graph-theoretic condition based on the union multigraph and the notion of a $\pi$-graph. The existence of a $\pi_p$-graph, which is a temporally evolving $\pi$-graph,  was shown to guarantee $\mathcal{SS}$ herdability. Future work will address extensions to temporal networks with arbitrary topologies.

\bibliographystyle{unsrt}        
\bibliography{autobib}           
\end{document}